\RequirePackage{amsmath}

\let\oldvec\vec
\documentclass{llncs}
\let\vec\oldvec

\usepackage{amssymb}
\usepackage{bitset}
\usepackage{multicol}
\usepackage{tikz}
\usetikzlibrary{arrows.meta}
\usepackage{subfig}

\usepackage{commands-tam}
\usepackage{graphicx}
\usepackage{url}
\usepackage{wrapfig}
\usepackage{cite}

\usepackage[textsize=scriptsize]{todonotes} 

\newcommand{\calS}{\mathcal{S}}

\newcommand{\calG}{\mathcal{G}}

\newcommand{\calW}{\mathcal{W}}
\newcommand\calU{\mathcal{U}}

\newcommand{\planter}{\texttt{planter}}
\newcommand{\iteration}{\texttt{iteration}}
\newcommand{\decrementer}{\texttt{decremenber}}

\newcommand{\geometrictile}[8]{
  \begin{tikzpicture}[x=#1, y=#1]
    \pgfmathsetmacro{\VA}{1 / (#2 + 2)}
    \pgfmathsetmacro{\VB}{1 - \VA}
    \bitsetSetBin{bsA}{#5}
    \bitsetSetBin{bsB}{#6}
    \bitsetSetBin{bsC}{#7}
    \bitsetSetBin{bsD}{#8}
    \filldraw[fill=#3] (\VA,\VA) -- (\VA,\VB) -- (\VB,\VB) -- (\VB,\VA) -- (\VA,\VA);
    \foreach \x in {1,...,#2}
    {
      \ifthenelse{\bitsetGet{bsC}{\x-1}=1}{
        \filldraw[fill=#4] (\VA*\x,0) -- (\VA*\x,\VA) -- (\VA*\x+\VA,\VA) -- (\VA*\x+\VA,0) -- (\VA*\x,0);}{}
      \ifthenelse{\bitsetGet{bsA}{\x-1}=1}{
        \filldraw[fill=#4] (\VA*\x,\VB) -- (\VA*\x,1) -- (\VA*\x+\VA,1) -- (\VA*\x+\VA,\VB) -- (\VA*\x,\VB);}{}
      \ifthenelse{\bitsetGet{bsD}{\x-1}=1}{
        \filldraw[fill=#4] (0,\VA*\x) -- (0,\VA*\x+\VA) -- (\VA,\VA*\x+\VA) -- (\VA,\VA*\x) -- (0,\VA*\x);}{}
      \ifthenelse{\bitsetGet{bsB}{\x-1}=1}{
        \filldraw[fill=#4] (\VB,\VA*\x) -- (\VB,\VA*\x+\VA) -- (1,\VA*\x+\VA) -- (1,\VA*\x) -- (\VB,\VA*\x);}{}
    }
  \end{tikzpicture}
}

\newcommand{\fourDomGeom}[9]{
    \bitsetSetBin{bsA}{#3}
    \bitsetSetBin{bsB}{#5}
    \bitsetSetBin{bsC}{#7}
    \bitsetSetBin{bsD}{#9}
    
    \foreach \x in {1,...,#1}
    {
        \ifthenelse{\bitsetGet{bsA}{#1-\x}=1}
        {\filldraw[fill=orange!50!white] (\x-1,2) -- (\x-1,3) -- (\x,3) -- (\x,2) -- (\x-1,2);}
        {\draw[dashed] (\x-1,2) -- (\x-1,3) -- (\x,3) -- (\x,2) -- (\x-1,2);}
    }
    \foreach \x in {1,...,#1}
    {
        \ifthenelse{\bitsetGet{bsB}{#1-\x}=1}
        {\filldraw[fill=orange!50!white] (\x+#1-1,2) -- (\x+#1-1,3) -- (\x+#1,3) -- (\x+#1,2) -- (\x+#1-1,2);}
        {\draw[dashed] (\x+#1-1,2) -- (\x+#1-1,3) -- (\x+#1,3) -- (\x+#1,2) -- (\x+#1-1,2);}
    }
    \foreach \x in {1,...,#1}
    {
        \ifthenelse{\bitsetGet{bsC}{#1-\x}=1}
        {\filldraw[fill=orange!50!white] (\x+2*#1-1,2) -- (\x+2*#1-1,3) -- (\x+2*#1,3) -- (\x+2*#1,2) -- (\x+2*#1-1,2);}
        {\draw[dashed] (\x+2*#1-1,2) -- (\x+2*#1-1,3) -- (\x+2*#1,3) -- (\x+2*#1,2) -- (\x+2*#1-1,2);}
    }
    \foreach \x in {1,...,#1}
    {
        \ifthenelse{\bitsetGet{bsD}{#1-\x}=1}
        {\filldraw[fill=orange!50!white] (\x+3*#1-1,2) -- (\x+3*#1-1,3) -- (\x+3*#1,3) -- (\x+3*#1,2) -- (\x+3*#1-1,2);}
        {\draw[dashed] (\x+3*#1-1,2) -- (\x+3*#1-1,3) -- (\x+3*#1,3) -- (\x+3*#1,2) -- (\x+3*#1-1,2);}
    }
	
	\draw [draw=red, line width=0.2mm] (0,3.2) -- (0,3.5) -- (#1-0.05,3.5) -- (#1-0.05,3.2);
	\draw [draw=blue, line width=0.2mm] (#1+0.05,3.2) -- (#1+0.05,3.5) -- (2*#1-0.05,3.5) -- (2*#1-0.05,3.2);
	\draw [draw=green!50!black, line width=0.2mm] (2*#1+0.05,3.2) -- (2*#1+0.05,3.5) -- (3*#1-0.05,3.5) -- (3*#1-0.05,3.2);
	\draw [draw=magenta, line width=0.2mm] (3*#1+0.05,3.2) -- (3*#1+0.05,3.5) -- (4*#1,3.5) -- (4*#1,3.2);
	
	\node at (0.5*#1,4) {#2};
	\node at (1.5*#1,4) {#4};
	\node at (2.5*#1,4) {#6};
	\node at (3.5*#1,4) {#8};
	
	\fill[fill=orange!50!white] (0,0) -- (0,2) -- (4*#1,2) -- (4*#1,0) -- (0,0);
	\draw (0,0) -- (0,2) -- (4*#1,2) -- (4*#1,0);
}

\let\OLDthebibliography\thebibliography
\renewcommand\thebibliography[1]{
  \OLDthebibliography{#1}
  \setlength{\parskip}{0pt}
  \setlength{\itemsep}{0pt plus 0.3ex}
}

\title{Geometric Tiles and Powers and Limitations of Geometric Hindrance in Self-Assembly}

\author{
 Daniel Hader \thanks{Department of Computer Science and Computer Engineering, University of Arkansas, Fayetteville, AR, USA. This author's research was supported in part by National Science Foundation Grants CCF-1422152 and CAREER-1553166. \protect\url{dhader@email.uark.edu}}
\and
 Matthew J. Patitz \thanks{Department of Computer Science and Computer Engineering, University of Arkansas, Fayetteville, AR, USA. This author's research was supported in part by National Science Foundation Grants CCF-1422152 and CAREER-1553166. \protect\url{patitz@uark.edu}}
}

\date{}

\institute{}

\begin{document}

\maketitle

\begin{abstract}
Tile-based self-assembly systems are capable of universal computation and algorithmically-directed growth. Systems capable of such behavior typically make use of ``glue cooperation'' in which the glues on at least $2$ sides of a tile must match and bind to those exposed on the perimeter of an assembly for that tile to attach. However, several models have been developed which utilize ``weak cooperation'', where only a single glue needs to bind but other preventative forces (such as geometric, or steric, hindrance) provide additional selection for which tiles may attach, and where this allows for algorithmic behavior. In this paper we first work in a model where tiles are allowed to have geometric bumps and dents on their edges. We show how such tiles can simulate systems of square tiles with complex glue functions (using asymptotically optimal sizes of bumps and dents), and also how they can simulate weakly cooperative systems in a model which allows for duples (i.e. tiles either twice as long or twice as tall as square tiles). We then show that with only weak cooperation via geometric hindrance, no system in any model can simulate even a class of tightly constrained, deterministic cooperative systems, further defining the boundary of what is possible using this tool.
\end{abstract}

\section{Introduction}

Systems of tile-based self-assembly in models such as the abstract Tile Assembly Model (aTAM)\cite{Winf98} have been shown to be very powerful in the sense that they are computationally universal\cite{Winf98} and are also able to algorithmically build complex structures very efficiently\cite{RotWin00,SolWin07}. The key to their computational and algorithmic power arises from the ability of tiles to convey information via the glues that they use to bind to growing assemblies and the preferential binding of some types of tiles over others based upon the requirement that they simultaneously match the glues of multiple tiles already in an assembly. This is called (glue) cooperation, and in physical implementations it can require a difficult balance of conditions to enforce. It is conjectured that systems which do not utilize cooperation, that is, those in which tiles can bind to a growing assembly by matching only a single glue, do not have the power to perform computations or algorithmically guided growth\cite{jLSAT1,WoodsMeunierSTOC,IUNeedsCoop}. However, several past results have shown that a middle ground, which we call weak cooperation, can be used to design systems which are capable of at least some of the power of cooperative systems. It has been shown that using geometric hindrance \cite{GeoTiles,jDuples,Polygons,Polyominoes} or repulsive glue forces \cite{SingleNegative,jNegativeGluesShapes,jBreakableDuples}, systems with a binding threshold of $1$ (a.k.a. temperature-1 systems) are capable of universal computation. This is because they are able to simulate temperature-2 \emph{zig-zag} aTAM systems, which are in many ways the most restrictive and deterministic of aTAM systems, but which are still capable of simulating arbitrary Turing machines.

In this paper, we further explore some of the powers and limitations of self-assembly in weakly-cooperative systems. First, we investigate the abilities of so-called geometric tiles (those with bumps and dents on their edges), which were shown in \cite{GeoTiles} to be able to self-assemble $n\times n$ squares using only temperature-1 and $\Theta(\sqrt{\log{n}})$ unique tile types (beating the lower bound of $\log{n}/\log{\log{n}}$ required for square aTAM tiles), and also at temperature-1 to be able to simulate temperature-2 zig-zag systems in the aTAM, and thus arbitrary Turing machines. Here we prove their ability to simulate non-cooperative, temperature-1, aTAM systems that have complex glue functions which allow glue types to bind to arbitrary sets of other glue types. We provide a construction and then show that it uses the minimum possible number of unique glue types (which is 2), and that it is asymptotically optimal with respect to the size of the geometries used. Next we show that another set of systems that are weakly-cooperative and use \emph{duples} (i.e. $2 \times 1$ sized tiles)\cite{jDuples} along with square tiles at temperature-1 can be simulated by geometric tiles.

Our final contribution is to expose a fundamental limitation of weakly cooperative self-assembling systems which rely on geometric hindrance. As previously mentioned, they are able to simulate the behaviors of temperature-2 aTAM systems which are called zig-zag systems. These systems have the properties that at all points during assembly, there is exactly one frontier location where the next tile can attach (or zero once assembly completes), and for every tile type, every tile of that type which attaches does so by using the exact same input sides (i.e. the initial sides with which it binds), and also has the same subset of sides used as output sides (i.e. sides to which later tiles attach).\footnote{Note that \emph{rectilinear} systems can also be simulated, and they have similar properties except that they may have multiple frontier locations available.} It has previously been shown in \cite{jDuples} that there exist temperature-2 aTAM systems which cannot be simulated by temperature-1 systems with duples, but that proof fundamentally utilizes a nondeterministic, undirected aTAM system with an infinite number of unique terminal assemblies. Here we try to find a ``tighter'' gap and so we explore aTAM temperature-2 systems which are directed and only ever have a single frontier location and whose tile types always have fixed input sides, but we make the slight change of allowing for tiles of the same type to sometimes use different output sides. We prove that with this minimal addition of uncertainty, no system which utilizes weak cooperation with geometric hindrance can simulate such a system, at any scale factor. Thus, while geometric hindrance is an effective tool for allowing the simulation of general computation, the dynamics which such weakly-cooperative systems can capture is severely restricted.

\section{Preliminaries}

Here we provide brief descriptions of models used in this paper. References are provided for more thorough definitions. 

\subsection{Informal description of the abstract Tile Assembly Model}
\label{sec:tam-informal}

This section gives a brief informal sketch of the abstract Tile Assembly Model (aTAM) \cite{Winf98} and uses notation from \cite{RotWin00} and \cite{jSSADST}. For more formal definitions and additional notation, see \cite{RotWin00} and \cite{jSSADST}.

A \emph{tile type} is a unit square with four sides, each consisting of a \emph{glue label} which is often represented as a finite string.
There is a finite set $T$ of tile types, but an infinite number of copies of each tile type, with each copy being referred to as a \emph{tile}.
A \emph{glue function} is a symmetric mapping from pairs of glue labels to a non-negative integer value which represents the strength of binding between those glues.
An \emph{assembly}
is a positioning of tiles on the integer lattice $\Z^2$, described  formally as a partial function $\alpha:\Z^2 \dashrightarrow T$. 
Let $\mathcal{A}^T$ denote the set of all assemblies of tiles from $T$, and let $\mathcal{A}^T_{< \infty}$ denote the set of finite assemblies of tiles from $T$.
We write $\alpha \sqsubseteq \beta$ to denote that $\alpha$ is a \emph{subassembly} of $\beta$, which means that $\dom\alpha \subseteq \dom\beta$ and $\alpha(p)=\beta(p)$ for all points $p\in\dom\alpha$.
Two adjacent tiles in an assembly \emph{interact}, or are \emph{attached}, if the glue labels on their abutting sides have positive strength between them according to the glue function. 
Each assembly induces a \emph{binding graph}, a grid graph whose vertices are tiles, with an edge between two tiles if they interact.
The assembly is \emph{$\tau$-stable} if every cut of its binding graph has strength at least~$\tau$, where the strength   of a cut is the sum of all of the individual glue strengths in the cut.


A \emph{tile assembly system} (TAS) is a 4-tuple $\calT = (T,\sigma,G,\tau)$, where $T$ is a finite set of tile types, $\sigma:\Z^2 \dashrightarrow T$ is a finite, $\tau$-stable \emph{seed assembly},
$G$ is a \emph{glue function},
and $\tau$ is the \emph{temperature}.
In the case that the glue function $G$ is \emph{diagonal}, meaning that each glue only has a non-zero strength with itself, $G$ is often omitted from the definition and a TAS is defined as the triple $\calT = (T, \sigma, \tau)$ where the strengths between identical glues are given as part of $T$. Glue functions which are not diagonal are often said to define \emph{flexible} glues.
Given an assembly $\alpha$, the \emph{frontier}, $\frontiert{\alpha}$, is the set of locations to which tiles can $\tau$-stably attach.
An assembly $\alpha$ is \emph{producible} if either $\alpha = \sigma$ or if $\beta$ is a producible assembly and $\alpha$ can be obtained from $\beta$ by the stable binding of a single tile to a location in $\frontiert{\beta}$.
In this case we write $\beta\to_1^\calT \alpha$ (to mean~$\alpha$ is producible from $\beta$ by the attachment of one tile), and we write $\beta\to^\calT \alpha$ if $\beta \to_1^{\calT*} \alpha$ (to mean $\alpha$ is producible from $\beta$ by the attachment of zero or more tiles).
We let $\prodasm{\calT}$ denote the set of producible assemblies of $\calT$.
An assembly $\alpha$ is \emph{terminal} if no tile can be $\tau$-stably attached to it, i.e. $|\frontiert{\alpha}| = 0$.
We let   $\termasm{\calT} \subseteq \prodasm{\calT}$ denote  the set of producible, terminal assemblies of $\calT$.
A TAS $\calT$ is \emph{directed} if $|\termasm{\calT}| = 1$. Hence, although a directed system may be nondeterministic in terms of the order of tile placements,  it is deterministic in the sense that exactly one terminal assembly. We say that a system $\calT$ is a \emph{single-assembly-sequence} system (SASS), if for every producible assembly $\alpha \in \prodasm{\calT}$, $|\frontiert{\alpha}| \le 1$, i.e. there is never more than one location to which a new tile can bind. If a system $\calT$ is a SASS and also directed, then it is fully deterministic.
We say that a system $\calT$ is a \emph{zig-zag system} if it is a SASS where, for every producible assembly $\alpha \in \prodasm{\calT}$ and $\beta \sqsubseteq \alpha$, the $y$ coordinate of $\frontiert{\alpha}$ is never smaller than the $y$ coordinate of $\frontiert{\beta}$.

\subsection{Informal description of the Geometric Tile Assembly Model}

\begin{figure}
    \label{fig:compatible-tiles}
    \begin{center}
        \geometrictile{2cm}{5}{white!60!cyan}{white!60!blue}{10110}{00101}{10011}{10011}
        \begin{tikzpicture}[x=2cm, y=1cm, baseline={([yshift=-6ex]current bounding box.center)}]
            \coordinate (p1) at (0,.5);
            \coordinate (p2) at (1,.5);
            \draw[draw=green!80!black, line width=2pt,{Latex[length=2mm]}-{Latex[length=2mm]}] (p1) -- node[above]{compatible} (p2);
        \end{tikzpicture}
        \geometrictile{2cm}{5}{white!60!cyan}{white!60!blue}{10110}{1111}{10011}{11010}
    \end{center}
    
    \begin{center}
        \geometrictile{2cm}{5}{white!60!cyan}{white!60!blue}{10110}{10111}{10011}{10011}
        \begin{tikzpicture}[x=2cm, y=1cm, baseline={([yshift=-8ex]current bounding box.center)}, node distance=0.8cm]
            \coordinate (p1) at (0,.5);
            \coordinate (p2) at (1,.5);
            \draw[draw=red, line width=2pt,{Latex[length=2mm]}-{Latex[length=2mm]}] (p1) -- node[above](A){compatible} node[above of=A]{not}  (p2);
        \end{tikzpicture}
        \geometrictile{2cm}{5}{white!60!cyan}{white!60!blue}{10110}{1111}{10011}{11010}
    \end{center}
    \caption{Examples of compatible and incompatible geometric tiles}
    \label{fig:geomTiles}
\end{figure}
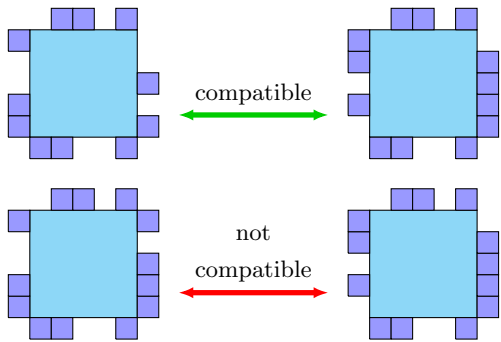

The geometric tile assembly model (GTAM) is similar to the aTAM with the addition of geometric bumps along the sides of tiles and the restriction that the glue function be diagonal. This section will provide an informal introduction to the model, but a more complete introduction can be found in \cite{GeoTiles}. The introduction presented here differs slightly from that in \cite{GeoTiles} in that we focus only on 1-dimensional geometry and try to match the notation as closely as possible to the aTAM definition in the previous section.

A \emph{geometry} of size $n$ is a mapping from $\{1,\ldots,n\}$ to $\{0,1\}$. This represents $n$ possible locations for bumps with a $1$ representing a bump at that location and a $0$ representing no bump. A \emph{geometric tile type} is a unit square with a glue label and a geometry on each side. Similarly to tiles in the aTAM, two geometric tiles \emph{interact} if the glue labels on their abutting sides have a positive strength; however, if the two abutting geometries have bumps in corresponding locations, they are called \emph{incompatible} and cannot bind regardless of glue strength. It's important to note that opposite sides of a tile can posses the same geometry. Geometry is rotated along the sides of the tile, so in this case, the geometry on the opposite side would be reversed. Thus any geometry which has a pair of bumps symmetric about its middle would be incompatible with itself. A \emph{geometric tile assembly system} (GTAS) is a triple $\calT = (T,\sigma,\tau)$, where $T$ is a finite set of geometric tile types, $\sigma:\Z^2 \dashrightarrow T$ is a finite, $\tau$-stable \emph{seed assembly}, and $\tau$ is the \emph{temperature}. Since the glue function for a GTAM system is diagonal, it is omitted from the definition for convenience. Also note that the size of the geometries in any GTAS is fixed.

\subsection{Additional Models}

A wide variety of models which generalize and extend certain aspects of the aTAM have been developed. Of those, due to space constraints we will briefly mention a few and cite references which can be used to find full definitions.

The 3D aTAM \cite{CookFuSch11} is the natural extension of the aTAM from 2-dimensional square tiles to 3-dimensional cubes. In \cite{jDuples}, the Dupled abstract Tile Assembly Model (DaTAM) is defined as an extension to the aTAM which allows for not only the standard square, $1 \times 1$, tiles, but also the inclusion of ``duples'' (or dominoes) which are tiles of dimension $2 \times 1$ or $1 \times 2$.  Allowing for more complex tile shapes, the Polyomino Tile Assembly Model (polyTAM) \cite{Polyominoes} allows for tiles composed of arbitrary numbers of unit squares which are connected along aligned faces. The polygonal TAM \cite{Polygons,OneTile} allows for tiles to have arbitrary polygonal shapes, and this is the only model mentioned which doesn't have an underlying regular lattice. Finally, another extension to the aTAM which we discuss is one which includes \emph{negative glues}, or glues which exhibit repulsive forces \cite{SingleNegative}. In systems with negative glues, two tiles may have adjacent faces with matching glues whose interaction strength is a negative integer. This is subtracted from the overall sum of binding strengths of adjacent glues to determine if the tile can attach.\footnote{Note that with negative glues, more complex dynamics, which include the breaking apart of assemblies, are possible\cite{jNegativeGluesShapes}.}

\subsection{Cooperation}
Self-assembling systems in tile assembly models contain a parameter known as the \emph{minimum binding threshold}, often called the \emph{temperature}. This parameter specifies the minimum biding strength, summed over all binding glues, that a tile must have with an assembly in order to attach to it. The binding strengths of glues are typically discrete, positive integer values. If the temperature parameter is set to $2$ or greater, we say that the system is \emph{strongly-cooperative}, uses \emph{strong cooperation}, or uses \emph{glue cooperation}, because it is possible for the attachment of a new tile to an assembly to require that it bind its glues, with positive affinity, to more than one tile already existing in the assembly. For example, in a temperature-2 system (i.e. one where the temperature parameter $= 2$), a tile may attach by binding with two glues each of strength 1, and thus two tiles already in the assembly \emph{cooperate} to allow for the new attachment. In contrast, we say that a system is \emph{non-cooperative} if its temperature parameter is set to 1 and in all situations where there is an empty location with an exposed incident glue, any tile with a matching glue can attach there. Finally, we say a system is \emph{weakly-cooperative}, or uses \emph{weak-cooperation}, if its temperature parameter is set to $1$ but it is able to make use of any form of \emph{binding hindrance}. In such a system, the binding strength of any single glue is strong enough to allow a tile to attach to an assembly; however, it is possible that some tiles are prevented from binding by another factor. Two such forms of binding hindrance we'll consider are \emph{geometric hindrance} and \emph{glue repulsion}. In this context, geometric (a.k.a. steric) hindrance occurs when a tile cannot bind in a location because at least some of the space that it would occupy is already occupied by some portion of another tile. This is relevant when tiles have more complex shapes than unit squares. See Figure~\ref{fig:geom-zigzag} for an example, as well as \cite{GeoTiles,Polygons,Polyominoes,jDuples}. Glue repulsion occurs when glues are able to experience negative strength interactions, i.e. when they can repel each other and their interaction subtracts from the total binding strength of a tile to an assembly. Examples can be found in \cite{SingleNegative,DotKarMasNegativeJournal,NegativeGluesShapes}. An example of the combination of geometric hindrance and glue repulsion can be found in \cite{jBreakableDuples}.

\begin{figure}
    \begin{center}
        \begin{tikzpicture}[x=0.4cm, y=0.4cm]
            \foreach \y in {0,...,2}{
            \foreach \x in {0,...,5}{
            \filldraw[fill=orange!50!white] (\x,\y)--(\x+1,\y)--(\x+1,\y+1)--(\x,\y+1)--(\x,\y);
            }}
            \filldraw[fill=orange!50!white] (5,3)--(6,3)--(6,4)--(5,4)--(5,3);
            \filldraw[fill=orange!50!white] (4,3)--(5,3)--(5,4)--(4,4)--(4,3);
            \draw[draw=red, line width=0.5mm, -latex] (0.5,0.5)--(5.5,0.5)--(5.5,1.5)--(0.5,1.5)--(0.5,2.5)--(5.5,2.5)--(5.5,3.5)--(4.2,3.5);
            
            \node at (0,7) {a)};
            
            \draw (7,-1)--(7,7);
            
            \fill[fill=orange!50!white] (8,0)--(8,1)--(13,1)--(13,3)--(16,3)--(16,0)--(8,0);
            \fill[fill=blue!50!white] (13,1.5)--(13.5,1.5)--(13.5,2.5)--(13,2.5)--(13,1.5);
            \fill[fill=red!50!white] (11.5,1)--(11.5,0.5)--(12.5,0.5)--(12.5,1)--(11.5,1);
            \draw (8,1)--(16,1) (9,1)--(9,0) (11,1)--(11,0) (13,3)--(13,0) (13,3)--(16,3) (15,3)--(15,0);
            \draw[dashed] (11,1)--(11,3)--(13,3);
            
            \fill[fill=orange!50!white] (8,4)--(10,4)--(10,6)--(8,6)--(8,4);
            \fill[red!50!white] (8.5,4)--(8.5,4.5)--(9.5,4.5)--(9.5,4)--(8.5,4);
            \fill[blue!50!white] (10,4.5)--(9.5,4.5)--(9.5,5.5)--(10,5.5)--(10,4.5);
            \draw (8,4)--(10,4)--(10,6)--(8,6)--(8,4);
            \draw[draw=green!60!black, line width=1mm] (8,6) circle (0.2cm);
            \fill[fill=orange!50!white] (11,4)--(13,4)--(13,6)--(11,6)--(11,4);
            \fill[white] (11.5,4)--(11.5,4.5)--(12.5,4.5)--(12.5,4)--(11.5,4);
            \fill[blue!50!white] (13,4.5)--(12.5,4.5)--(12.5,5.5)--(13,5.5)--(13,4.5);
            \draw (11,4)--(13,4)--(13,6)--(11,6)--(11,4);
            \draw[draw=red, line width=1mm] (10.7,5.7)--(11.3,6.3) (10.7,6.3)--(11.3,5.7);
            \fill[fill=orange!50!white] (14,4)--(16,4)--(16,6)--(14,6)--(14,4);
            \fill[red!50!white] (14.5,4)--(14.5,4.5)--(15.5,4.5)--(15.5,4)--(14.5,4);
            \fill[white] (16,4.5)--(15.5,4.5)--(15.5,5.5)--(16,5.5)--(16,4.5);
            \draw (14,4)--(16,4)--(16,6)--(14,6)--(14,4);
            \draw[draw=red, line width=1mm] (13.7,5.7)--(14.3,6.3) (13.7,6.3)--(14.3,5.7);
            
            \node at (8,7) {b)};
            
            \draw (17,-1)--(17,7);
            
            \fill[fill=orange!50!white] (21.5,1) arc (180:0:0.2cm);
            \fill[fill=orange!50!white] (18,0)--(18,1)--(23,1)--(23,3)--(26,3)--(26,0)--(18,0);
            \fill[fill=blue!50!white] (23,1.5)--(23.5,1.5)--(23.5,2.5)--(23,2.5)--(23,1.5);
            \draw (18,1)--(21.5,1) (21.5,1) arc (180:0:0.2cm) (22.5,1)--(26,1) (19,1)--(19,0) (21,1)--(21,0) (23,3)--(23,0) (23,3)--(26,3) (25,3)--(25,0);
            \draw[dashed] (21,1)--(21,3)--(23,3);
            
            \fill[fill=orange!50!white] (18,4)--(18.5,4) arc (180:0:0.2cm)--(20,4)--(20,6)--(18,6)--(18,4);
            \fill[blue!50!white] (20,4.5)--(19.5,4.5)--(19.5,5.5)--(20,5.5)--(20,4.5);
            \draw (18,4)--(18.5,4) arc (180:0:0.2cm)--(20,4)--(20,6)--(18,6)--(18,4);
            \draw[draw=green!60!black, line width=1mm] (18,6) circle (0.2cm);
            \fill[fill=orange!50!white] (21,4)--(21.7,4)--(21.7,4.5)--(22.3,4.5)--(22.3,4)--(23,4)--(23,6)--(21,6)--(21,4);
            \fill[blue!50!white] (23,4.5)--(22.5,4.5)--(22.5,5.5)--(23,5.5)--(23,4.5);
            \draw (21,4)--(21.7,4)--(21.7,4.5)--(22.3,4.5)--(22.3,4)--(23,4)--(23,6)--(21,6)--(21,4);;
            \draw[draw=red, line width=1mm] (20.7,5.7)--(21.3,6.3) (20.7,6.3)--(21.3,5.7);
            \fill[fill=orange!50!white] (24,4)--(24.5,4) arc (180:0:0.2cm)--(26,4)--(26,6)--(24,6)--(24,4);
            \fill[white] (26,4.5)--(25.5,4.5)--(25.5,5.5)--(26,5.5)--(26,4.5);
            \draw (24,4)--(24.5,4) arc (180:0:0.2cm)--(26,4)--(26,6)--(24,6)--(24,4);
            \draw[draw=red, line width=1mm] (23.7,5.7)--(24.3,6.3) (23.7,6.3)--(24.3,5.7);
            
            \node at (18,7) {c)};
            
        \end{tikzpicture}
    \end{center}
    \caption{(a) An illustration of how a zig-zag system propagates upward, snaking east and west. (b) Each of the new tile additions can be thought of as acting with 2 inputs. In a temperature-2 system, as in most zig-zag systems, the bottom and side inputs come from cooperating glues and only the tile that matches both can grow. (c) In temperature-1 GTAM systems there is no cooperation so the side input uses a glue while the bottom geometry is used to prevent the wrong tiles from growing.}
    \label{fig:geom-zigzag}
\end{figure}
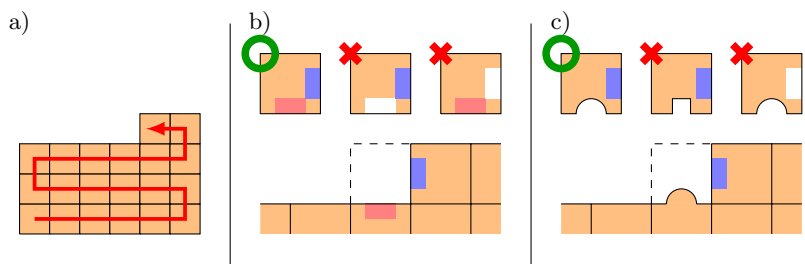

\subsection{Simulation}

Here we give a very brief intuitive definition of what it means for one tile assembly system to simulate another. See Section~\ref{sec:simulation_def_append} for more technically detailed definitions related to simulation, especially as it relates to scale factors greater than 1.

Intuitively, simulation of a system $\calT$ by a system $\calS$ requires that there is some scale factor $m \in \Z^+$ such that $m \times m$ squares of tiles in $\calS$ represent individual tiles in $\calT$, and there is a ``representation function'' capable of inspecting assemblies in $\calS$ and mapping them to assemblies in $\calT$. A representation function $R$ takes as input an assembly in $\calS$ and returns an assembly in $\calT$ to which it maps. In order for $\calS$ to correctly simulate $\calT$, it must be the case that for any producible assembly $\alpha \in \prodasm{\calT}$ that there is a corresponding assembly $\beta \in \prodasm{\calS}$ such that $R(\beta) = \alpha$. (Note that there may be more than one such $\beta$.) Furthermore, for any $\alpha' \in \prodasm{\calT}$ which can result from a tile addition to $\alpha$, there exists $\beta' \in \prodasm{\calS}$ which can result from the addition of one or more tiles to $\beta$, and conversely, $\beta$ can only grow into assemblies which can be mapped into valid assemblies of $\calT$ into which $\alpha$ can grow.

\section{Simulation of Temperature-1 aTAM Systems}

\begin{theorem}
	For any temperature-1 aTAM system, $\calT=(\T, \sigma_\calT, G_\calT, 1)$, with arbitrary symmetric glue function (a.k.a. flexible glues), there exists a temperature-1 GTAM system $\calU=(\U, \sigma_\calU, 1)$ that simulates $\calT$ using only 2 distinct glues, tile geometries of size $4n$, where $n$ is the number of glues in $\calT$, and scale factor 1.
	\label{thm:atam-sim}
\end{theorem}

\begin{wrapfigure}{r}{0.5\textwidth}
    \vspace{-20pt}
    \begin{center}
        \begin{tikzpicture}[x=1.4cm, y=1.4cm]
            \fill[fill=gray!50!white] (0,0) -- (1,0) -- (1,1) -- (0,1) -- (0,0);
            \fill[fill=gray!50!white] (1,0) -- (2,0) -- (2,1) -- (1,1) -- (1,0);
            \fill[fill=gray!50!white] (0,1) -- (1,1) -- (1,2) -- (0,2) -- (0,1);
            \fill[fill=red!50!white] (1,1) -- (0.5,1.5) -- (1,2) -- (1,1);
            \fill[fill=green!50!white] (1,1) -- (1.5,0.5) -- (2,1) -- (1,1);
            \draw (0,0) -- (1,0) -- (1,1) -- (0,1) -- (0,0);
            \draw (1,0) -- (2,0) -- (2,1) -- (1,1) -- (1,0);
            \draw (0,1) -- (1,1) -- (1,2) -- (0,2) -- (0,1);
            
            \fill[fill=gray!50!white] (1.5,2.5) -- (2.5,2.5) -- (2.5,3.5) -- (1.5,3.5) -- (1.5,2.5);
            \fill[fill=red!50!white] (1.5,2.5) -- (2,3) -- (1.5,3.5) -- (1.5,2.5);
            \fill[fill=blue!50!white] (1.5,2.5) -- (2,3) -- (2.5,2.5) -- (1.5,2.5);
            \draw (1.5,2.5) -- (2.5,2.5) -- (2.5,3.5) -- (1.5,3.5) -- (1.5,2.5);
            
            \fill[fill=gray!50!white] (2.5,1) -- (3.5,1) -- (3.5,2) -- (2.5,2) -- (2.5,1);
            \fill[fill=blue!50!white] (2.5,1) -- (3,1.5) -- (2.5,2) -- (2.5,1);
            \fill[fill=green!50!white] (2.5,1) -- (3,1.5) -- (3.5,1) -- (2.5,1);
            \draw (2.5,1) -- (3.5,1) -- (3.5,2) -- (2.5,2) -- (2.5,1);
            
            \draw[color=red!80!white, line width=2pt, -latex] (1.8,2.4) -- (1.5,1.6);
            \draw[color=red!80!white, line width=2pt, -latex] (2.4,1.5) -- (1.6,1.5);
            
        \end{tikzpicture}
    \end{center}
    \caption{An example of a situation in an aTAM system which mandates the use of two glues in the simulating GTAM system. Here the blue glue is incompatible with both the red and green glues. If only a single glue was used in the GTAM system, both of the tiles would necessarily be incompatible and could not fit.}
    \vspace{-10pt}
    \label{fig:1glueIssue}
\end{wrapfigure}
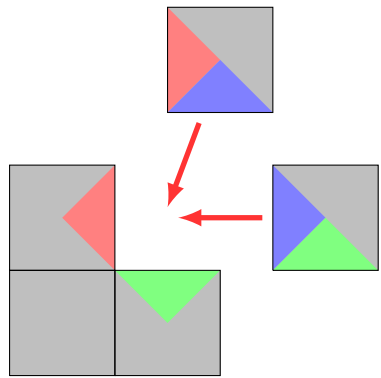

To prove this theorem, as is done fully in section \ref{sec:proof-atam-sim}, we construct a GTAM system which simulates a given aTAM system. The construction seeks to implement the binding behaviour of glues using the compatibility behaviour of geometries. To represent any given glue we construct two corresponding geometries which we call the $\alpha$ and $\beta$ versions. Both of these geometries are divided into $4$ domains of size $n$, meaning that there are $4n$ potential bump locations on each. The domains, from left to right, are named $\alpha_1$, $\beta_1$, $\beta_2$, and $\alpha_2$ and examples of what these domains look like can be seen in figure \ref{fig:glues2geom}. Also keep in mind that geometries are rotated in order to be placed on the various faces of a geometric tile. This means that the westernmost domain of a north geometry on a geometric tile would be $\alpha_1$, whereas the westernmost domain of a south geometry would be $\alpha_2$. When indexing bump locations, we count from left to right for domains $\alpha_1$ and $\beta_1$ and count from right to left for domains $\alpha_2$ and $\beta_2$. This means that if two geometries were on abutting faces of adjacent tiles, the location $i$, for $1\le i \le n$, of domain $\alpha_1$ on the first geometry would line up with location $i$ of domain $\alpha_2$ on the second geometry.

\begin{wrapfigure}{r}{0.5\textwidth}
    \vspace{-30pt}
    \begin{center}
		$$G =
		\begin{bmatrix}
			0 & 1 & 0 & 1 \\
			1 & 1 & 0 & 0 \\
			0 & 0 & 1 & 0 \\
			1 & 0 & 0 & 1 \\
		\end{bmatrix}$$

        \begin{tikzpicture}[x=.35cm, y=.35cm]
		\fourDomGeom{4}{$\alpha_1$}{1000}{$\beta_1$}{0000}{$\beta_2$}{0000}{$\alpha_2$}{0101}
		\node at (8, 1) {geometry: $\gamma_{1\alpha}$};
		\end{tikzpicture}
		\begin{tikzpicture}[x=.35cm, y=.35cm]
		\fourDomGeom{4}{$\alpha_1$}{0000}{$\beta_1$}{1000}{$\beta_2$}{0101}{$\alpha_2$}{0000}
		\node at (8, 1) {geometry: $\gamma_{1\beta}$};
		\end{tikzpicture}
		\begin{tikzpicture}[x=.35cm, y=.35cm]
		\fourDomGeom{4}{$\alpha_1$}{0100}{$\beta_1$}{0000}{$\beta_2$}{0000}{$\alpha_2$}{1100}
		\node at (8, 1) {geometry: $\gamma_{2\alpha}$};
		\end{tikzpicture}
		\begin{tikzpicture}[x=.35cm, y=.35cm]
		\fourDomGeom{4}{$\alpha_1$}{0000}{$\beta_1$}{0100}{$\beta_2$}{1100}{$\alpha_2$}{0000}
		\node at (8, 1) {geometry: $\gamma_{2\beta}$};
		\end{tikzpicture}
    \end{center}
    \vspace{-10pt}
	\caption{A glue function can be represented by a symmetric matrix. The strength of the bond between glues $g_i$ and $g_j$ is represented by the value $G_{i,j}$. Given a glue function, we can construct geometries whose compatibility behaviour emulates the binding behaviour of the glues. Illustrated above are the $\alpha$ and $\beta$ versions of the geometries corresponding to the glues $g_1$ and $g_2$ as described by the glue function.}
	\vspace{-10pt}
	\label{fig:glues2geom}
\end{wrapfigure}
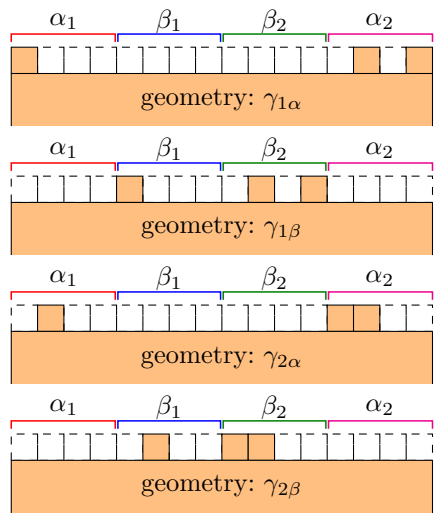

Out of these 4 domains, each geometry has only 2 functional domains. The $\alpha$ version of a geometry will only use the $\alpha$ domains and the $\beta$ version of a geometry will only use the $\beta$ domains. The first functional domain in each of these geometries, the $\alpha_1$ and $\beta_1$ domains respectively, encode which glue is being represented by placing a bump in the corresponding location. For example, if a geometry corresponds to glue 3, it will have a bump in location 3 of its first functional domain. The second domains in each of the two geometries, $\alpha_2$ and $\beta_2$ respectively, encode the binding behaviour of the corresponding glue. This is done by placing a bump in all of the locations corresponding to glues to which the represented glue cannot interact. For example, if the represented glue does not interact with glues 1 and 3, then there will be bumps in locations 1 and 3 of the second functional domain. Consider the $\alpha$ versions of the geometries corresponding to two glues, say $g_1$ and $g_2$, which cannot interact. The $\alpha_2$ domain of the geometry corresponding to $g_1$ will have a bump in position $2$ indicating that it cannot interact with glue $g_2$. This bump will be incompatible with the bump in location $2$ of the $\alpha_1$ domain of the geometry corresponding to $g_2$. The same is true for the $\beta$ domains of the $\beta$ version geometries. so, whenever two glues cannot interact, the corresponding geometries of the same version will be incompatible. Figure \ref{fig:glues2geom} demonstrates a glue function and what some of the corresponding geometries look like.

The reason that we need 2 versions of each geometry is to accommodate mismatches. Because mismatched glues can, and often do, legitimately occur in aTAM systems, like in figure \ref{fig:1glueIssue}, we need two versions, $\alpha$ and $\beta$, of each geometry which are always compatible with geometries of the other version. Because geometries of different versions use exclusive functional domains, they will always be geometrically compatible with one another. The $\alpha$ domains occur on the outside of a geometry and the $\beta$ domains on the inside, so when abutting, they cannot overlap. Moreover, we need 2 distinct glues in $\calU$; since, while we do want geometries to be compatible with opposite versions, we don't want them initiating growth with tiles whose geometries are of opposite version. Thus all $\alpha$ version geometries will have one glue and all $\beta$ version geometries another. This allows the $\alpha$ and $\beta$ versions of glues to represent mismatches in the simulated system.

\begin{wrapfigure}{r}{0.35\textwidth}
    \vspace{-10pt}
    \begin{center}
        \begin{tikzpicture}[x=0.5cm, y=0.5cm]
            
            \draw[draw=cyan!50!white, line width=1mm] (1,1) -- (1,4);
            \draw[draw=orange!50!white, line width=1mm] (1,1) -- (4,1);
            \draw[draw=green!50!white, line width=1mm] (4,1) -- (4,7);
            \draw[draw=red!50!white, line width=1mm] (1,4) -- (7,4);
            \draw[draw=blue!50!white, line width=1mm] (7,1) -- (7,4);
            \draw[draw=blue!50!white, line width=1mm] (1,7) -- (4,7);
        
            \fill[fill=gray!70!white] (0,0) -- (0,2) -- (2,2) -- (2,0) -- (0,0);
            \fill[fill=cyan!50!white] (2,2) -- (1,1) -- (0,2) -- (2,2);
            \fill[fill=orange!50!white] (2,2) -- (1,1) -- (2,0) -- (2,2);
            \draw (0,0) -- (0,2) -- (2,2) -- (2,0) -- (0,0);
            
            \fill[fill=gray!70!white] (3,0) -- (3,2) -- (5,2) -- (5,0) -- (3,0);
            \fill[fill=green!50!white] (5,2) -- (4,1) -- (3,2) -- (5,2);
            \fill[fill=orange!50!white] (3,2) -- (4,1) -- (3,0) -- (3,2);
            \draw (3,0) -- (3,2) -- (5,2) -- (5,0) -- (3,0);
            
            \fill[fill=gray!70!white] (0,3) -- (2,3) -- (2,5) -- (0,5) -- (0,3);
            \fill[fill=red!50!white] (2,5) -- (1,4) -- (2,3) -- (2,5);
            \fill[fill=cyan!50!white] (2,3) -- (1,4) -- (0,3) -- (2,3);
            \draw (0,3) -- (2,3) -- (2,5) -- (0,5) -- (0,3);
            
            \fill[fill=gray!70!white] (6,0) -- (6,2) -- (8,2) -- (8,0) -- (6,0);
            \fill[fill=blue!50!white] (6,2) -- (7,1) -- (8,2) -- (6,2);
            \draw (6,0) -- (6,2) -- (8,2) -- (8,0) -- (6,0);
            
            \fill[fill=gray!70!white] (6,3) -- (6,5) -- (8,5) -- (8,3) -- (6,3);
            \fill[fill=blue!50!white] (6,3) -- (7,4) -- (8,3) -- (6,3);
            \fill[fill=red!50!white] (6,3) -- (7,4) -- (6,5) -- (6,3);
            \draw (6,3) -- (6,5) -- (8,5) -- (8,3) -- (6,3);
            
            \fill[fill=gray!70!white] (0,6) -- (2,6) -- (2,8) -- (0,8) -- (0,6);
            \fill[fill=blue!50!white] (2,6) -- (1,7) -- (2,8) -- (2,6);
            \draw (0,6) -- (2,6) -- (2,8) -- (0,8) -- (0,6);
            
            \fill[fill=gray!70!white] (3,6) -- (5,6) -- (5,8) -- (3,8) -- (3,6);
            \fill[fill=blue!50!white] (3,6) -- (4,7) -- (3,8) -- (3,6);
            \fill[fill=green!50!white] (3,6) -- (4,7) -- (5,6) -- (3,6);
            \draw (3,6) -- (5,6) -- (5,8) -- (3,8) -- (3,6);
            
            \node at (1,1) {\textbf{S}};
            \node at (1,4) {\textbf{U}};
            \node at (4,1) {\textbf{R}};
            \node at (4,7) {\textbf{A}};
            \node at (7,4) {\textbf{B}};
            
        \end{tikzpicture}
    \end{center}
    \caption{The tile set used in Theorem \ref{thm:single-glue}. The lines between tiles represent possible attachments. Notice that, if tile \emph{S} is the seed, the final configuration must be a 2 by 2 square.}
    \vspace{-10pt}
    \label{fig:1glueTileset}
\end{wrapfigure}
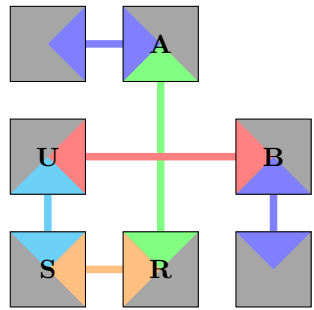

This construction demonstrates that the behaviour of temperature-1 aTAM systems with arbitrary symmetric glue functions can be simulated by temperature-1 GTAM systems with no cost in scale using a fixed number of glues, namely $2$. It's important to note that the GTAM systems in our construction only use standard glues which bind only to themselves. This proof implies that two glues in a GTAM system are sufficient for simulating arbitrary aTAM systems at temperature-1; the following theorem shows that, using fewer than two glues, not all aTAM systems can be simulated by GTAM systems. This implies that two glues are necessary for allowing glue mismatches to be properly simulated.

\begin{theorem}
     There exists an aTAM system at temperature-1 which cannot be simulated by a temperature-1 GTAM system at scale factor 1 using $< 2$ glues, regardless of the geometry size of the tiles.
     \label{thm:single-glue}
\end{theorem}

\begin{proof} Consider the temperature 1 aTAM system, say $\calT$, presented in Figure \ref{fig:1glueTileset}, wherein the tile labelled \emph{S} is the seed. Notice that $\calT$ is not directed; there are multiple final configurations, each of which are $2 \times 2$ squares. Also, let it be the case that each glue binds only to itself, so that the blue glue does not bind with any other glue, for example. Now, for contradiction, suppose that there is some temperature-1 GTAM system $\calU$ that does simulate $\calT$ at scale factor 1 using only a single glue. Since $\calU$ simulates $\calT$, it must be able to simulate the growth of $\calT$ from any possible configuration. Consider then the configuration of $\calT$ in which the tiles labelled \emph{U} and \emph{R} have grown to form an \emph{L} shape. In this configuration, there are two tiles which can grow into the corner opposite to tile \emph{S}: tiles \emph{A} and \emph{B}. Furthermore, notice that if either tile attaches, there will be some glue mismatch since the blue glue does not match with either the red or green glues.

Now imagine a corresponding, \emph{L} shaped configuration in $\calU$. Since $\calU$ simulates $\calT$, there must be some geometric tile corresponding to either tile \emph{A} or \emph{B} which can attach in the corner opposite \emph{S}. If we suppose, without loss of generality, that this was a tile corresponding to \emph{A}, then it must be the case that this geometric tile has a geometry on its west face which is compatible with the geometry on the east face of a tile corresponding to \emph{U} despite the fact that the glues don't match in $\calT$. Therefore, in the case where a geometric tile corresponding to tile \emph{R} hasn't yet attached, there would be the possibility for the geometric tile corresponding to \emph{U} to attach to that geometric tile corresponding to \emph{A} since there is only a single strength-1 glue and the geometries are compatible. This, however, would be a violation of the dynamics of $\calT$ and thus such a $\calU$ could not simulate $\calT$. \qed
\end{proof}

The previous proof demonstrated that a GTAM system needs at least two glues to simulate aTAM systems at temperature-1 and scale factor 1. Furthermore, the prior proof gave a construction of a GTAM system which used exactly two glues to simulate any aTAM system at temperature-1 and scale factor 1. Additionally, the construction used geometries of size $4n$, and it is shown in \cite{GeoTiles} that the lower bound on the size of geometries needed to represent some non-diagonal glue functions is $\Theta(n)$.

\section{Simulation of Temperature-1 Duple TAM Systems}

The Dupled abstract Tile Assembly Model (DaTAM) is similar to the aTAM, but with the addition of duple tiles which are simply 2$\times$1 tiles. While this addition may seem minimal, the DaTAM allows for weak cooperation via geometric hindrance, and is capable of universal computation at temperature-1 \cite{jDuples}. It's important to note that in the DaTAM, we assume a diagonal glue function.

\begin{theorem}
    For every temperature-1 DaTAM system $\calD$, there exists a temp\-erature-1 GTAM system $\calS$, using only 2 glues, which simulates it at scale factor 1.
    \label{thm:duple-sim}
\end{theorem}

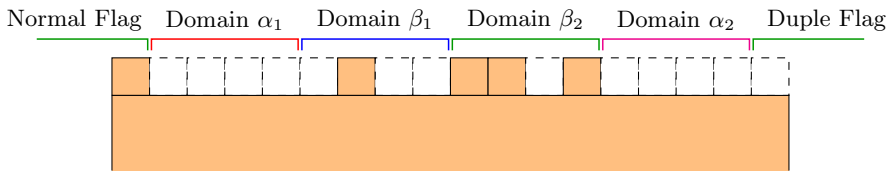
\begin{figure}
    \begin{center}
        \begin{tikzpicture}[x=.5cm, y=.5cm]
    		\def\drawsquare[#1]#2#3{\filldraw[#1] (#2,#3) -- (#2,#3 + 1) -- (#2 + 1,#3 + 1) -- (#2 + 1,#3) -- (#2,#3);}
    		\drawsquare [fill=orange!50!white]{0}{2}
    		
    		\drawsquare [fill=white, dashed]{1}{2}
    		\drawsquare [fill=white, dashed]{2}{2}
    		\drawsquare [fill=white, dashed]{3}{2}
    		\drawsquare [fill=white, dashed]{4}{2}
    		
    		\drawsquare [fill=white, dashed]{5}{2}
    		\drawsquare [fill=orange!50!white]{6}{2}
    		\drawsquare [fill=white, dashed]{7}{2}
    		\drawsquare [fill=white, dashed]{8}{2}
    		
    		\drawsquare [fill=orange!50!white]{9}{2}
    		\drawsquare [fill=orange!50!white]{10}{2}
    		\drawsquare [fill=white, dashed]{11}{2}
    		\drawsquare [fill=orange!50!white]{12}{2}
    		
    		\drawsquare [fill=white, dashed]{13}{2}
    		\drawsquare [fill=white, dashed]{14}{2}
    		\drawsquare [fill=white, dashed]{15}{2}
    		\drawsquare [fill=white, dashed]{16}{2}
    		
    		\drawsquare [fill=white, dashed]{17}{2}
    		
    		\draw [draw=red, line width=0.2mm] (1.05,3.2) -- (1.05,3.5) -- (4.95,3.5) -- (4.95,3.2);
    		\draw [draw=blue, line width=0.2mm] (5.05,3.2) -- (5.05,3.5) -- (8.95,3.5) -- (8.95,3.2);
    		\draw [draw=green!60!black, line width=0.2mm] (9.05,3.2) -- (9.05,3.5) -- (12.95,3.5) -- (12.95,3.2);
        	\draw [draw=magenta, line width=0.2mm] (13.05,3.2) -- (13.05,3.5) -- (16.95,3.5) -- (16.95,3.2);
    		\draw [draw=green!60!black, line width=0.2mm] (-2,3.5) -- (0.95,3.5) -- (0.95,3.2);
    		\draw [draw=green!60!black, line width=0.2mm] (17.05,3.2) -- (17.05,3.5) -- (20,3.5);
    		
    		\node at (-1,4) {Normal Flag};
    		\node at (3,4) {Domain $\alpha_1$};
    		\node at (7,4) {Domain $\beta_1$};
    		\node at (11,4) {Domain $\beta_2$};
    		\node at (15,4) {Domain $\alpha_2$};
    		\node at (19,4) {Duple Flag};
    		
    		\fill[fill=orange!50!white] (0,0) -- (0,2) -- (18,2) -- (18,0) -- (0,0);
    		\draw (0,0) -- (0,2) -- (18,2) -- (18,0);
    	\end{tikzpicture}
    \end{center}
	\caption{Example of geometry corresponding to glue from a DaTAM system. Notice that the normal flag contains a bump and not the duple flag suggesting that this geometry corresponds to a normal glue and not the glue that would go between the two tiles that make up a duple tile.}
	\label{fig:duple_geometry}
\end{figure}

The proof of this theorem, presented fully in appendix section \ref{sec:proof-duple-sim}, is very similar to the aTAM simulation construction above. The only difference is the addition of two extra bump locations on each geometry. These bump locations appear on the far left and far right of the geometry and represent whether the geometry represents a normal glue or a special duple geometry respectively. Normal glues behave exactly as in the aTAM construction. Since no duple tiles exist in the GTAM, such tiles have to be simulated by half tiles in the GTAM. Two half tiles make a duple and the geometry between them is a duple geometry. Notice that duple geometries and normal glue geometries will be incompatible because their bumps are opposite each other, meaning that when the geometries lie on abutting faces of adjacent tiles, the bumps will be overlapping. Because of this, if half of a duple attaches, no other tile, except the correct second half, will be able to grow into the location reserved for the second half since the bumps would interfere. Moreover, a duple will never attach to a position wherein its second half is already blocked for the same reason. This enforces that the dynamics of the GTAM system can model those of the DaTAM system.
\section{Glue Cooperation Cannot be Simulated With Geometric Hindrance}

In this section, we first show that there exists a directed temperature-2 aTAM SASS (i.e. a fully deterministic temperature-2 aTAM system) which cannot be simulated by any temperature-1 GTAM system. A brief overview is given here, and the full proof can be found in the Appendix.

\begin{theorem}\label{thm:imposs}
There exists a directed temperature-2 aTAM SASS $\calS$ that cannot be simulated by any GTAM system at temperature 1.
\end{theorem}

Figure~\ref{fig:imposs-overview} shows a high-level, schematic drawing of the system $\calS$ which cannot be simulated. Essentially, it grows a ``$\planter$'' module (similar to that of \cite{jCCSA}) to form an infinite assembly growing to the right, which initiates an infinite series of counters which grow upward to every height $\ge 4$. (Figure~\ref{fig:imposs-mid} shows the pattern of growth which allows $\calS$ to be a SASS.) Each counter then grows an arm down which crashes into a portion of the assembly below it, but since the arms grow longer and longer, eventually they reach a point where they must ``pump'', or grow in a periodic manner.  However, in order to correctly grow macrotiles which simulate the cooperative growth between the end of each arm and the bottom portion of the assembly, there must be path of tiles which can grow out from each arm. Since the arms must become periodic, those paths could also grow in higher locations, which leads to invalid simulation. Examples can be seen in Figures~\ref{fig:imposs-zoom} and \ref{fig:imposs-windows}.

\begin{figure}
    \centering
    \includegraphics[width=4.2in]{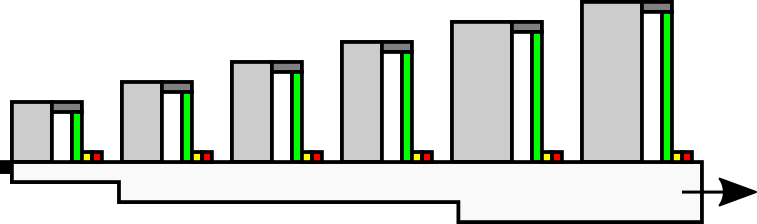}
    \caption{Overview of a temperature-2 aTAM system which cannot be simulated by any temperature-1 GTAM system at scale factor 1.}
    \label{fig:imposs-overview}
\end{figure}

\begin{figure}[htp]
\centering
    \subfloat{
        \label{fig:imposs-mid}
        \includegraphics[width=1.7in]{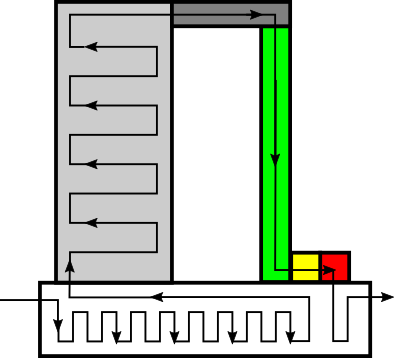}
        }
  \quad
  \subfloat{
        \label{fig:imposs-zoom}
        \includegraphics[width=1.2in]{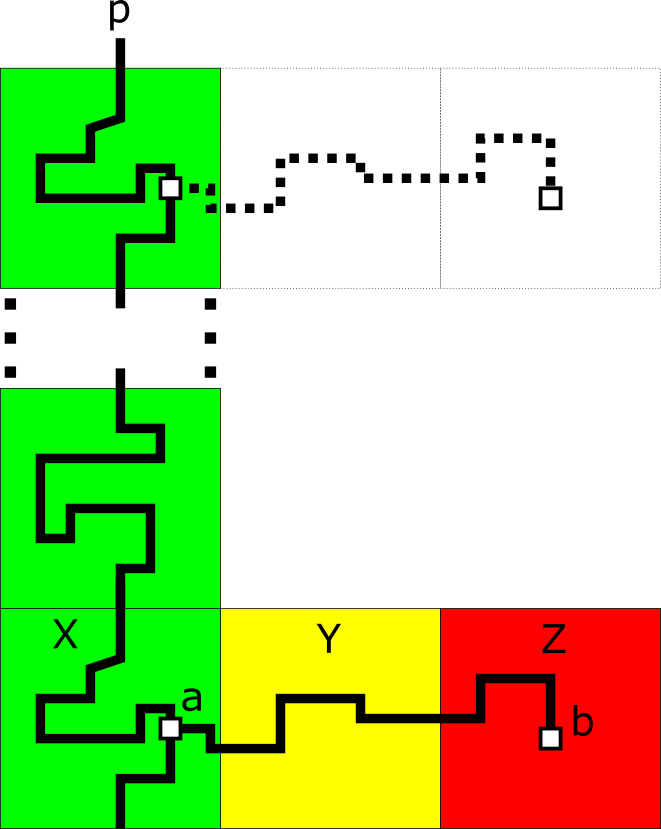}
        }
  \caption{(a) Depiction of one iteration of the growth of $\calT$, with arrows showing the ordering of growth, (b) Zoomed in portion of the construction shown in Figure~\ref{fig:imposs-overview} which shows (with a solid line) an example of a path of tiles bound by glues which must extend from a tile, $a$, in the supertile representing a green tile, to a tile, $b$, in the supertile representing the red tile. The dashed line shows how a previous copy of $a$ could allow growth of the same path in a higher location.}
  \label{fig:types_of_vertices}
\end{figure}

\begin{figure}
    \centering
    \includegraphics[width=4.8in]{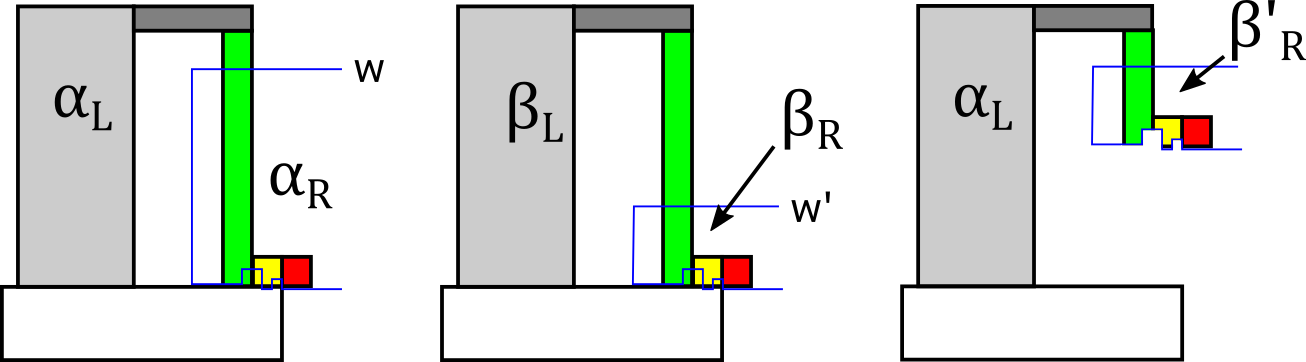}
    \caption{(left and middle) Examples of windows, $w$ and $w'$ which each cut a portion of the supertiles representing the green column, plus the yellow and red tiles, from the rest of an $\iteration$, and which have identical glue bindings across them. Note that glue bindings only occur across the top line of each, and the bottom line separating the inside from the $\planter$ goes only between unbound tiles, (right) Assembly $\alpha_L\beta_R'$ (where $\beta_R'$ is simply a translated copy of $\beta_R$) which must be able to form by the Window Movie Lemma. Even if the representation of the red tile isn't complete, the allowed boundary for the growth of fuzz is broken.}
    \label{fig:imposs-windows}
\end{figure}

While Theorem~\ref{thm:imposs} states that temperature-1 GTAM systems cannot even simulate the full class of directed temperature-2 aTAM single-assembly-sequence systems, the following result generalizes that to show that the same is true across all systems relying on weak cooperation across any tile assembly model.

\begin{theorem}\label{thm:imposs-general}
There exists a directed temperature-2 aTAM SASS $\calS$ that cannot be simulated by any weakly-cooperative tile assembly system that relies on geometric hindrance. 
\end{theorem}

The proof of Theorem~\ref{thm:imposs-general} is essentially identical to that of Theorem~\ref{thm:imposs}.

Note that Theorem~\ref{thm:imposs-general} is proven for weakly-cooperative systems using geometric hindrance, but that this does not include the second category of types of binding hindrance, namely systems which use glue repulsion. Although such systems make possible dynamic behavior in which portions of an assembly may break off, which may make the proof more difficult, we conjecture that they also cannot simulate $\calS$.

\bibliographystyle{splncs04}
\bibliography{tam,experimental_refs}

\clearpage

\section*{Technical Appendix}
This section contains technical details of definitions and proofs which are omitted from the main body due to space constraints.

\subsection{Simulation}
\label{sec:simulation_def_append}

This section contains a formal, rigorous definition of what it means for one tile assembly system to ``simulate'' another. Our definitions come from \cite{IUNeedsCoop}, but we make slight modifications to account for the simulation of geometric tiles. Also, note that a great amount of the complexity required for the definitions arises due to the possible dynamics of simulations with scale factors $> 1$, and that otherwise the mapping of assemblies and equivalence of production and dynamics are much more straightforward.


From this point on, let $T$ be a tile set, and let $m\in\Z^+$.
An \emph{$m$-block supertile} over $T$ is a partial function $\alpha : \Z_m^2 \dashrightarrow T$, where $\Z_m = \{0,1,\ldots,m-1\}$.
Let $B^T_m$ be the set of all $m$-block supertiles over $T$.
The $m$-block with no domain is said to be $\emph{empty}$.
If $T$ consists of square tiles, for a general assembly $\alpha:\Z^2 \dashrightarrow T$ and $(x,y)\in\Z^2$, define $\alpha^m_{x,y}$ to be the $m$-block supertile defined by $\alpha^m_{x,y}(i_x, i_y) = \alpha(mx+i_x, my+i_y)$ for $0 \leq i_x,i_y< m$. If $T$ consists of geometric tiles, then additional space is used to represent the geometry regions of tiles. See Figures~\ref{fig:GTAM-grid} and \ref{fig:GTAM-macrotiles} for a depiction, and instead define $\alpha^m_{x,y}$ to be the $m$-block supertile defined by $\alpha^m_{x,y}(i_x, i_y) = \alpha(mx+nx+i_x, my+ny+i_y)$ for $0 \leq i_0,i_1< m$.
For some tile set $S$, a partial function $R: B^{S}_m \dashrightarrow T$ is said to be a \emph{valid $m$-block supertile representation} from $S$ to $T$ if for any $\alpha,\beta \in B^{S}_m$ such that $\alpha \sqsubseteq \beta$ and $\alpha \in \dom R$, then $R(\alpha) = R(\beta)$.

\begin{figure}
    \centering
    \includegraphics[width=1.5in]{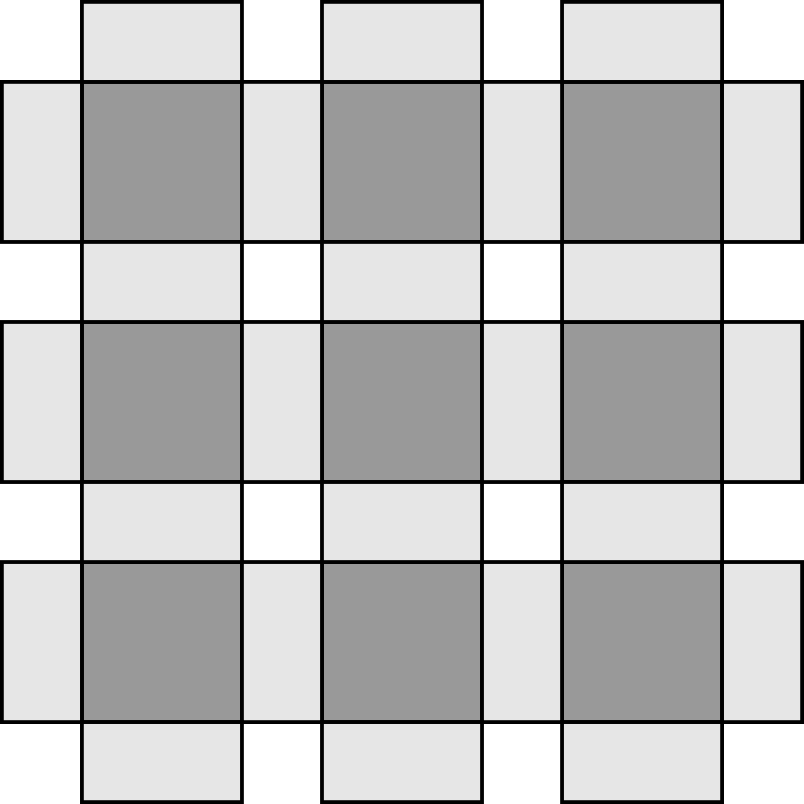}
    \caption{The grid formed by GTAM tiles, with the tile bodies shown in dark grey and the geometry regions, which overlap for adjacent tiles, shown in light grey.}
    \label{fig:GTAM-grid}
\end{figure}

\begin{figure}
    \centering
    \includegraphics[width=1.5in]{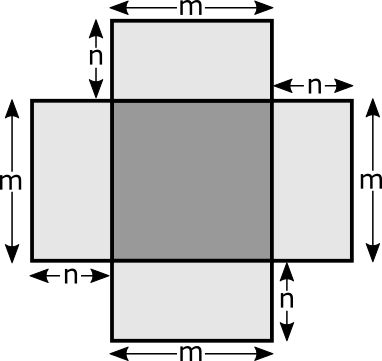}
    \caption{A macrotile for the simulation of a GTAM tile. The tile body is represented by an $m \times m$ square and the geometry regions by $n \times n$ rectangles.}
    \label{fig:GTAM-macrotiles}
\end{figure}

For a given valid $m$-block supertile representation function $R$ from tile set~$S$ to tile set $T$, define the \emph{assembly representation function}\footnote{Note that $R^*$ is a total function since every assembly of $S$ represents \emph{some} assembly of~$T$; the functions $R$ and $\alpha$ are partial to allow undefined points to represent empty space.}  $R^*: \mathcal{A}^{S} \rightarrow \mathcal{A}^T$ such that $R^*(\alpha') = \alpha$ if and only if $\alpha(x,y) = R\left(\alpha'^m_{x,y}\right)$ for all $(x,y) \in \Z^2$.
For an assembly $\alpha' \in \mathcal{A}^{S}$ such that $R(\alpha') = \alpha$, $\alpha'$ is said to map \emph{cleanly} to $\alpha \in \mathcal{A}^T$ under $R^*$ if for all non empty blocks $\alpha'^m_{x,y}$, $(x,y)+(u_x,u_y) \in \dom \alpha$ for some $u_x,u_y \in \{0,1\}$ such that $u_x^2 + u_y^2 \leq 1$, or if $\alpha'$ has at most one non-empty $m$-block~$\alpha^m_{0, 0}$. In other words, $\alpha'$ may have tiles on supertile blocks representing empty space in $\alpha$, but only if that position is adjacent to a tile in $\alpha$.  We call such growth ``around the edges'' of $\alpha'$ \emph{fuzz} and thus restrict it to be adjacent to only valid supertiles, but not diagonally adjacent (i.e.\ we do not permit \emph{diagonal fuzz}). If $T$ consists of geometric tiles, the fuzz may also be in the regions around the macrotiles which represent the geometry regions. 


In the following definitions, let $\mathcal{T} = \left(T,\sigma_T,\tau_T\right)$ be a tile assembly system, let $\mathcal{S} = \left(S,\sigma_S,\tau_S\right)$ be a tile assembly system, and let $R$ be an $m$-block representation function $R:B^S_m \rightarrow T$.

\begin{definition}
\label{def-equiv-prod} We say that $\mathcal{S}$ and $\mathcal{T}$ have \emph{equivalent productions} (under $R$), and we write $\mathcal{S} \Leftrightarrow \mathcal{T}$ if the following conditions hold:
\begin{enumerate}
        \item $\left\{R^*(\alpha') | \alpha' \in \prodasm{\mathcal{S}}\right\} = \prodasm{\mathcal{T}}$.
        \item $\left\{R^*(\alpha') | \alpha' \in \termasm{\mathcal{S}}\right\} = \termasm{\mathcal{T}}$.
        \item For all $\alpha'\in \prodasm{\mathcal{S}}$, $\alpha'$ maps cleanly to $R^*(\alpha')$.
\end{enumerate}
\end{definition}

\begin{definition}
\label{def-t-follows-s} We say that $\mathcal{T}$ \emph{follows} $\mathcal{S}$ (under $R$), and we write $\mathcal{T} \dashv_R \mathcal{S}$ if $\alpha' \rightarrow^\mathcal{S} \beta'$, for some $\alpha',\beta' \in \prodasm{\mathcal{S}}$, implies that $R^*(\alpha') \to^\mathcal{T} R^*(\beta')$.
\end{definition}

\begin{definition}
\label{def-s-models-t} We say that $\mathcal{S}$ \emph{models} $\mathcal{T}$ (under $R$), and we write $\mathcal{S} \models_R \mathcal{T}$, if for every $\alpha \in \prodasm{\mathcal{T}}$, there exists $\Pi \subset \prodasm{\mathcal{S}}$ where $R^*(\alpha') = \alpha$ for all $\alpha' \in \Pi$, such that, for every $\beta \in \prodasm{\mathcal{T}}$ where $\alpha \rightarrow^\mathcal{T} \beta$, (1) for every $\alpha' \in \Pi$ there exists $\beta' \in \prodasm{\mathcal{S}}$ where $R^*(\beta') = \beta$ and $\alpha' \rightarrow^\mathcal{S} \beta'$, and (2) for every $\alpha'' \in \prodasm{\mathcal{S}}$ where $\alpha'' \rightarrow^\mathcal{S} \beta'$, $\beta' \in \prodasm{\mathcal{S}}$, $R^*(\alpha'') = \alpha$, and $R^*(\beta') = \beta$, there exists $\alpha' \in \Pi$ such that $\alpha' \rightarrow^\mathcal{S} \alpha''$.
\end{definition}

The previous definition essentially specifies that every time $\mathcal{S}$ simulates an assembly $\alpha \in \prodasm{\mathcal{T}}$, there must be at least one valid growth path in $\mathcal{S}$ for each of the possible next steps that $\mathcal{T}$ could make from $\alpha$ which results in an assembly in $\mathcal{S}$ that maps to that next step.

\begin{definition}
\label{def-s-simulates-t} We say that $\mathcal{S}$ \emph{simulates} $\mathcal{T}$ (under $R$) if $\mathcal{S} \Leftrightarrow_R \mathcal{T}$ (equivalent productions), $\mathcal{T} \dashv_R \mathcal{S}$ and $\mathcal{S} \models_R \mathcal{T}$ (equivalent dynamics).
\end{definition}

\subsection{Proof of Theorem \ref{thm:atam-sim}} \label{sec:proof-atam-sim}

\begin{proof} To prove Theorem~\ref{thm:atam-sim}, we will begin with an arbitrary temperature-1 aTAM system $\calT=(\T, \sigma_\calT, G_\calT, 1)$ which has a symmetric glue function, and we will construct a GTAM system $\calU = (U, \sigma_{\calU}, 1)$ with the necessary constraints which simulates it at scale factor 1. The first step in this construction will consist of defining a number of geometries whose compatibility behaviour will emulate the binding behaviour of the glues in $\calT$. For purposes which will become clear later in the proof, we will define two geometries in $\calU$ for each glue in $\calT$. Let $\{g_1, \ldots, g_n\}$ be an enumeration of the glues used in $\calT$ with $n$ being the number of glues. For each glue $g_i$ in this enumeration, we will construct two geometries, $\gamma_{i\alpha}$ and $\gamma_{i\beta}$, both of size $4n$. These geometries will be divided into four contiguous domains of size $n$. Figure \ref{fig:glues2geom} provides an illustration of how the domains are laid out. The two outermost domains, labelled $\alpha_1$ and $\alpha_2$ from left to right, form what will be called the $\alpha$-domains; likewise, the two inner domains, $\beta_1$ and $\beta_2$, make up the $\beta$-domains. It's important to notice that, when lying on the abutting faces of two adjacent tiles (i.e. North and South, or East and West), the geometries would be reversed from each other, meaning that the $\alpha_1$ domain of one geometry would be overlapping the $\alpha_2$ domain of the other, and the $\beta_1$ domain of one would be overlapping the $\beta_2$ domain of the other. For this reason, when indexing the bumps from $1$ to $n$ in each domain, we will number the bump locations in domains $\alpha_1$ and $\beta_1$ from left to right and in domains $\alpha_2$ and $\beta_2$ from right to left. Thus, in figure \ref{fig:glues2geom}, for example, geometry $\gamma_{1\alpha}$ has a bump in location 1 in domain $\alpha_1$ and two bumps in domain $\alpha_2$ in locations 1 and 3.

Each domain in a geometry serves a purpose. Domains $\alpha_1$ and $\beta_1$ are domains whose bumps indicate the index of the glue being represented. For example, when representing glue $g_i$, the geometry $\gamma_{i\alpha}$ will have a single bump at location $i$ in domain $\alpha_1$ and the geometry $\gamma_{i\beta}$ will have a single bump at location $i$ in domain $\beta_1$. The domains $\alpha_2$ and $\beta_2$ are used to functionally indicate the binding behaviour of the glue being represented. A bump in either of these domains encodes the index of a glue which cannot bind with the glue being represented. For example, in figure \ref{fig:glues2geom}, the glue function $G$, represented as a matrix, indicates that glue $1$ cannot bind with either glue $1$ or $3$, thus the geometries corresponding to $g_1$, namely $\gamma_{1\alpha}$ and $\gamma_{1\beta}$, have bumps in domains $\alpha_2$ and $\beta_2$ respectively at locations $1$ and $3$. Notice that, because this glue has no strength with itself, if either of these geometries appeared on the abutting faces of two tiles, those tiles would be incompatible, since a bump at location 1 in domain $\alpha_1$ would intersect with a bump at location 1 in domain $\alpha_2$ and likewise for the $\beta$-domains.

Explicitly, for each glue $g_i$ in $\calT$, our construction of the geometries $\gamma_{i\alpha}$ and $\gamma_{i\beta}$ is as follows. Let both the domain $\alpha_1$ in $\gamma_{i\alpha}$ and the domain $\beta_1$ in $\gamma_{i\beta}$ contain a single bump at location $i$. Then for each $j$ such that $G_\calT(i, j)=0$, i.e. glues $i$ and $j$ do not bind with any strength, let domain $\alpha_2$ in $\gamma_{i\alpha}$ and domain $\beta_2$ in $\gamma_{i\beta}$ have a bump at location $j$, remembering to count from right to left in these domains. Everywhere else will contain no bumps in either geometry. Here, it's important to note that the $\alpha$-domains of $\gamma_{i\alpha}$ are identical to the $\beta$-domains of $\gamma_{i\beta}$ and neither contains any bumps in the opposite type domains. It may seem as though these geometries are redundant; however, as will be explained later, they become important when trying to simulate certain behaviours of some tile sets. Furthermore, it's not difficult to see that exactly when two glues in $\calT$ can bind, as per the glue function $G_\calT$, the corresponding geometries of a single type, either $\alpha$ or $\beta$, will be compatible with each other. Also notice that $\alpha$ and $\beta$ geometries will always be compatible with each other no matter the corresponding glues since they use non-overlapping domains.

Now let $g^\alpha$ and $g^\beta$ be the $2$ glues to be used in $\calU$, with the property that each binds to itself with strength 1 and not the other. For each tile $t_i \in \T$, we will construct 16 different tiles $u_i^1,\ldots,u_i^{16}$ in $\U$. Each of these 16 tiles will correspond to the single original tile for purposes of simulation, but, as will be seen later, multiple corresponding geometric tile types are necessary in order to allow for adjacent tiles to have mismatching glues, as can happen in situations like that depicted in figure \ref{fig:1glueIssue}. If $g_N$ is the glue that appears on the north face of $t_i$, then there are two corresponding geometries that can go on the north face of a corresponding geometric tile in $\U$, namely $\gamma_{N\alpha}$ and $\gamma_{N\beta}$. Since this is true for all four glues of $t_i$, there are $2^4=16$ possibilities for constructing geometric tiles that correspond to $t_i$. Moreover since each side of a geometric tile type needs a glue as well as a geometry, we assign the corresponding glue, $g^\alpha$ or $g^\beta$, to a side depending on whether it has the $\alpha$ or $\beta$ version of the geometry. These $16|\T|$ geometric tile types make up the tile set $\U$.

Because temperature-1 does not allow for cooperative binding between tiles, there are only a few cases that need to be considered to see that $\calU$ has equivalent dynamics to $\calT$. Consider the situation in which there is a tile $t_0$ which has a glue, say $g_0$, on its north face, north chosen without loss of generality, that admits the growth of one or more tiles, say $t_1,\ldots,t_m$, in $\calT$. Because $g_0$ is responsible for the growth of these tiles, they must have glues on their south faces which can bind with strength 1 to $g_0$. Given that $\calT$ is arbitrary, let these be called $g_1,\ldots,g_k$. Now consider the situation in $\calU$. We can assume that a tile $u_0$, corresponding to $t_0$, has already been placed and that the geometry on its north face is either $\gamma_{0\alpha}$ or $\gamma_{0\beta}$. Without loss of generality, we can assume $\gamma_{0\alpha}$, thus we know, by definition, that the glue on the north face of $u_0$ is $g^\alpha$. Because $u_0$ is in a temperature-1 GTAM system, growth happens exactly when there is a glue match and a compatible geometry. We know that the only geometries that are compatible with $\gamma_{0\alpha}$ are the $\alpha$ versions of those which correspond to $g_1,\ldots,g_k$, as this is how the geometries were defined, and any $\beta$ version geometry; however, because $u_0$ has glue $g^\alpha$ on its north face, only the $\alpha$ version geometries will be able to grow since any tile face with a $\beta$ version geometry will have glue $g^\beta$.

Therefore, the tile that grows to the north of $u_0$ will be a tile with an $\alpha$ geometry corresponding to one of the glues $g_1,\ldots,g_k$ on its south face. By definition, these are the tiles corresponding to $t_1,\ldots,t_m$ or, more precisely, the half of the corresponding geometric tiles with an $\alpha$ type geometry on their south side. If there is no other tile adjacent to this tile being placed, then nothing more needs to be considered, and a tile corresponding to a valid tile in $\calT$ can be placed north of $u_0$. However, when there is an adjacent tile to the tile being placed, as in the case of figure \ref{fig:1glueIssue}, then additional consideration is necessary. This is the reason why 16 tiles corresponding to each tile in $\calT$ were necessary. Consider the situation in figure \ref{fig:1glueIssue} where the tile with a single green glue represents $t_0$ and the tile with the single red glue represents a tile adjacent to the tile about to be placed. Even supposing that red glues mismatch with blue glues, the tile with a green and blue glue should be a valid tile to place to the north of $t_0$ since it matches at least one glue in $\calT$. In $\calU$ however, if we only had geometries of one type, the geometries corresponding to the red and blue glues would be incompatible and not allow the tile to be placed. Having a second glue type however, allows us to have a version of the tile which can be placed north of $u_0$ and mismatch on the left. Since the $\alpha$ and $\beta$ type geometries are always compatible with each other, such mismatched glues can always be placed adjacent to each other even though their glues don't match.

To complete the construction, given $\sigma_\calT$, it is straightforward to construct $\sigma_\calU$. For each tile in $\sigma_\calT$ simply replace it with a corresponding tile in $\calU$. Since there are 16 such tiles for each tile in $\sigma_\calT$, start with the tile using all $\alpha$ type geometries. Then, if there are any mismatches between adjacent geometries, choose, arbitrarily, one of the mismatched geometries to be a $\beta$ type. Since this does not affect the matching of any other geometries and since it fixes the problem for the two mismatched geometries, this procedure can be used to fix all mismatched geometries until there are none in $\sigma_\calU$. 

Finally, to show that $\calU$ simulates $\calT$, we first note that for each tile type $t_i \in T$, the representation function $R$ simply maps each of the $16$ tile types in $U$ which were explicitly created to represent it (with each permutation of $\alpha$ and $\beta$ sides) to $t_i$. It is then clear that $R(\sigma_\calU) = \sigma_\calT$, and for every tile addition which is possible to $\sigma_\calT$, and then to any producible assembly $\alpha \in \prodasm{\calT}$, tiles of exactly one of the $16$ tile types in $U$ which map to that tile can be added to the corresponding $\beta \in \prodasm{\calU}$ where $R(\beta) = \alpha$, and no other attachments are possible to $\beta$. Therefore, $\calU$ and $\calT$ have equivalent productions and dynamics, and $\calU$ simulates $\calT$. This simulation is done at scale factor 1, and $U$ has only $2$ distinct glues and geometries of size $4n$.
\qed
\end{proof}

\subsection{Proof of Theorem \ref{thm:duple-sim}} \label{sec:proof-duple-sim}

\begin{proof}
Given a DaTAM system $\calD$, the construction of a GTAM system $\calS$ which simulates it is very similar to the aTAM construction above with a few key differences. Since in the GTAM we are only allowed tiles that are unit squares (with small geometries on their sides), we will create unique tile types for each half of a duple tile type in $\calD$. Then, we include two bumps to the left and right of our four domains. These bumps are called flag bumps and they are used to distinguish between the geometry that represents normal glues and the geometry that adjoins the two halves of a duple tile. Second, since the DaTAM only has glues that bind with themselves, or in other words, only uses diagonal glue functions, the $\alpha_2$ or $\beta_2$ domain, whichever is not empty, of any geometry will only contain a single gap in its bumps corresponding to the bump in its first domain. This simplifies the construction a bit and allows us to focus on the interesting aspect of the DaTAM: the duples. Since the GTAM does not have the capacity for duple tiles we have to simulate the properties of duple tiles using geometry. The property of duple tiles we have to be particularly careful in emulating is the duple tile's ability to block other tiles from growing a tile away from where the duple itself attached to the system. This means that as soon as either half of a duple tile attaches to our assembly, it must be impossible for any tile other than the opposite half to grow into the adjacent space. Fortunately, geometry is capable of this.

When simulating a DaTAM system, the flag bump to the left of the $\alpha_1$ domain represents that the geometry corresponds to a normal glue in the DaTAM system. The flag bump to the right of the $\alpha_2$ domain represents that the geometry corresponds to the center of a duple tile. These flags do more than just distinguish geometry types, however. Notice that a geometry containing the duple flag bump will be incompatible with any geometry containing the normal flag bump. This means that even when only half of a duple has grown in, any normal tile will not be able to grow where the second half should because their geometries will be incompatible.

Let $G = \{g_1,\ldots,g_n\}$ be some enumeration of the glues in $\calD$ with $n$ being $|G|$ and let $D = \{d_1,\ldots,d_m\}$ be an enumeration of all duple tile types in $\calD$ with $m=|D|$. Now let $k = \max(m, n)$. For each $g_i$ in $G$, we construct two geometries $\gamma_{i\alpha}$ and $\gamma_{i\beta}$ of size $4k + 2$ in the exact same way as in the aTAM simulation construction above using four domains of size $k$, but with the addition of two bumps locations representing the flags on the outside of the geometry (We use $k$ instead of $n$ for the size of the domains because we will need to be able to represent all of the duple tiles as well). Because these geometries correspond to normal glues, the leftmost flag, the one corresponding to normal tiles, will be in place, meaning that there will be a bump in this first position. See figure \ref{fig:duple_geometry} for an example of what such a geometry might look like. For tiles which have no glue, or equivalently a glue of strength 0, on a side in $\calD$, we construct a geometry $\gamma_0$ which contains only the normal flag bump and is thus compatible with all other normal glue geometries. Furthermore, for each $d_i\in D$, we construct a geometry $\lambda_i$, also of size $4k+2$, as follows. The geometry will be divided into the same domains and flags as for the normal tiles; however, the leftmost normal flag bump will be missing and the rightmost duple flag bump will be present. The $\alpha_1$ domain will correspond to the index of the duple $i$ and the $\alpha_2$ domain will contain bumps in every position except for the one corresponding to $i$ counting backwards. Since such a glue will only appear in between the two halves corresponding to a duple tile, we don't have to make a $\beta$ version of this glue as there cannot be any mismatches. Also notice that even though there might be a geometry corresponding to a glue and a geometry corresponding to a duple with the same index, they will be incompatible geometries since the flag bumps do not match.

\begin{figure}
    \begin{center}
        \begin{tikzpicture}[x=0.3cm, y=0.3cm]
            \pgfmathsetmacro{\x}{0}
            \pgfmathsetmacro{\y}{0}
            
            \fill[fill=gray!70!white] (\x,\y) -- (\x,\y+5.4) -- (\x+1,\y+5.4) -- (\x+1,\y+1) -- (\x+3,\y+1) -- (\x+3,\y+3) -- (\x+5,\y+3) -- (\x+5,\y+1) -- (\x+5.4,\y+1) -- (\x+5.4,\y) -- (\x,\y);
            \fill[fill=green!50!white] (\x+1,\y+3.4) -- (\x+0.6,\y+3.4) -- (\x+0.6,\y+4.6) -- (\x+1,\y+4.6) -- (\x+1,\y+3.4);
            \fill[fill=yellow!50!white] (\x+3.4,\y+3) -- (\x+3.4,\y+2.6) -- (\x+4.6,\y+2.6) -- (\x+4.6,\y+3) -- (\x+3.4,\y+3);
            \draw (\x+1,\y+5.4) -- (\x+1,\y+1) -- (\x+3,\y+1) -- (\x+3,\y+3) -- (\x+5,\y+3) -- (\x+5,\y+1) -- (\x+5.4,\y+1);
            
            \pgfmathsetmacro{\x}{8}
            \pgfmathsetmacro{\y}{-2}
            
            \fill[fill=gray!70!white] (\x,\y) -- (\x,\y+5.4) -- (\x+1,\y+5.4) -- (\x+1,\y+1) -- (\x+3,\y+1) -- (\x+3,\y+3) -- (\x+5,\y+3) -- (\x+5,\y+1) -- (\x+5.4,\y+1) -- (\x+5.4,\y) -- (\x,\y);
            \fill[fill=green!50!white] (\x+1,\y+3.4) -- (\x+0.6,\y+3.4) -- (\x+0.6,\y+4.6) -- (\x+1,\y+4.6) -- (\x+1,\y+3.4);
            \fill[fill=yellow!50!white] (\x+3.4,\y+3) -- (\x+3.4,\y+2.6) -- (\x+4.6,\y+2.6) -- (\x+4.6,\y+3) -- (\x+3.4,\y+3);
            \draw (\x+1,\y+5.4) -- (\x+1,\y+1) -- (\x+3,\y+1) -- (\x+3,\y+3) -- (\x+5,\y+3) -- (\x+5,\y+1) -- (\x+5.4,\y+1);
            
            \fill[fill=blue!50!white] (\x+3,\y+3) -- (\x+3,\y+5) -- (\x+5,\y+5) -- (\x+5,\y+3) -- (\x+3,\y+3);
            \fill[fill=yellow!50!white] (\x+3.4,\y+3) -- (\x+3.4,\y+3.4) -- (\x+4.6,\y+3.4) -- (\x+4.6,\y+3) -- (\x+3.4,\y+3);
            \draw (\x+3,\y+3) -- (\x+3,\y+5) -- (\x+5,\y+5) -- (\x+5,\y+3) -- (\x+3,\y+3);
            
            \pgfmathsetmacro{\x}{8}
            \pgfmathsetmacro{\y}{4}
            
            \fill[fill=gray!70!white] (\x,\y) -- (\x,\y+5.4) -- (\x+1,\y+5.4) -- (\x+1,\y+1) -- (\x+3,\y+1) -- (\x+3,\y+3) -- (\x+5,\y+3) -- (\x+5,\y+1) -- (\x+5.4,\y+1) -- (\x+5.4,\y) -- (\x,\y);
            \fill[fill=green!50!white] (\x+1,\y+3.4) -- (\x+0.6,\y+3.4) -- (\x+0.6,\y+4.6) -- (\x+1,\y+4.6) -- (\x+1,\y+3.4);
            \fill[fill=yellow!50!white] (\x+3.4,\y+3) -- (\x+3.4,\y+2.6) -- (\x+4.6,\y+2.6) -- (\x+4.6,\y+3) -- (\x+3.4,\y+3);
            \draw (\x+1,\y+5.4) -- (\x+1,\y+1) -- (\x+3,\y+1) -- (\x+3,\y+3) -- (\x+5,\y+3) -- (\x+5,\y+1) -- (\x+5.4,\y+1);
            
            \fill[fill=red!50!white] (\x+1,\y+3) -- (\x+1,\y+5) -- (\x+5,\y+5) -- (\x+5,\y+3) -- (\x+1,\y+3);
            \fill[fill=green!50!white] (\x+1,\y+3.4) -- (\x+1.4,\y+3.4) -- (\x+1.4,\y+4.6) -- (\x+1,\y+4.6) -- (\x+1,\y+3.4);
            \draw (\x+1,\y+3) -- (\x+1,\y+5) -- (\x+5,\y+5) -- (\x+5,\y+3) -- (\x+1,\y+3);
            
            \draw[draw=red, line width=0.5mm, -latex] (5.5,2) -- (7.5,1);
            \draw[draw=red, line width=0.5mm, -latex] (5.5,4) -- (7.5,5);

            \pgfmathsetmacro{\x}{16}
            \pgfmathsetmacro{\y}{0}
            
            \fill[fill=gray!70!white] (\x,\y) -- (\x,\y+5.4) -- (\x+1,\y+5.4) -- (\x+1,\y+1) -- (\x+3,\y+1) -- (\x+3,\y+3) -- (\x+5,\y+3) -- (\x+5,\y+1) -- (\x+5.4,\y+1) -- (\x+5.4,\y) -- (\x,\y);
            \draw (\x+1,\y+5.4) -- (\x+1,\y+1) -- (\x+3,\y+1) -- (\x+3,\y+3) -- (\x+5,\y+3) -- (\x+5,\y+1) -- (\x+5.4,\y+1);
            
            \filldraw[fill=green!50!white] (\x+1,\y+4.8) -- (\x+1.2,\y+4.8) -- (\x+1.2,\y+4.6) -- (\x+1,\y+4.6) -- (\x+1,\y+4.8);
            \filldraw[fill=yellow!50!white] (\x+3.2,\y+3) -- (\x+3.2,\y+3.2) -- (\x+3.4,\y+3.2) -- (\x+3.4,\y+3) -- (\x+3.2,\y+3);
            
            \pgfmathsetmacro{\x}{24}
            \pgfmathsetmacro{\y}{-2}
            
            \fill[fill=gray!70!white] (\x,\y) -- (\x,\y+5.4) -- (\x+1,\y+5.4) -- (\x+1,\y+1) -- (\x+3,\y+1) -- (\x+3,\y+3) -- (\x+5,\y+3) -- (\x+5,\y+1) -- (\x+5.4,\y+1) -- (\x+5.4,\y) -- (\x,\y);
            \draw (\x+1,\y+5.4) -- (\x+1,\y+1) -- (\x+3,\y+1) -- (\x+3,\y+3) -- (\x+5,\y+3) -- (\x+5,\y+1) -- (\x+5.4,\y+1);
            
            \filldraw[fill=green!50!white] (\x+1,\y+4.8) -- (\x+1.2,\y+4.8) -- (\x+1.2,\y+4.6) -- (\x+1,\y+4.6) -- (\x+1,\y+4.8);
            \filldraw[fill=blue!50!white] (\x+3,\y+3.2) -- (\x+2.8,\y+3.2) -- (\x+2.8,\y+3.4) -- (\x+3,\y+3.4) -- (\x+3,\y+3.2);
            
            \filldraw[fill=blue!50!white] (\x+3,\y+3) -- (\x+3,\y+5) -- (\x+5,\y+5) -- (\x+5,\y+3) -- (\x+3,\y+3);
            
            \pgfmathsetmacro{\x}{24}
            \pgfmathsetmacro{\y}{4}
            
            \fill[fill=gray!70!white] (\x,\y) -- (\x,\y+5.4) -- (\x+1,\y+5.4) -- (\x+1,\y+1) -- (\x+3,\y+1) -- (\x+3,\y+3) -- (\x+5,\y+3) -- (\x+5,\y+1) -- (\x+5.4,\y+1) -- (\x+5.4,\y) -- (\x,\y);
            \draw (\x+1,\y+5.4) -- (\x+1,\y+1) -- (\x+3,\y+1) -- (\x+3,\y+3) -- (\x+5,\y+3) -- (\x+5,\y+1) -- (\x+5.4,\y+1);
            
            \filldraw[fill=red!50!white] (\x+3,\y+3.2) -- (\x+3.2,\y+3.2) -- (\x+3.2,\y+3.4) -- (\x+3,\y+3.4) -- (\x+3,\y+3.2);
            \filldraw[fill=yellow!50!white] (\x+3.2,\y+3) -- (\x+3.2,\y+3.2) -- (\x+3.4,\y+3.2) -- (\x+3.4,\y+3) -- (\x+3.2,\y+3);
            
            \filldraw[fill=red!50!white] (\x+1,\y+3) -- (\x+1,\y+5) -- (\x+3,\y+5) -- (\x+3,\y+3) -- (\x+1,\y+3);
            
            \draw[draw=red, line width=0.5mm, -latex] (21.5,2) -- (23.5,1);
            \draw[draw=red, line width=0.5mm, -latex] (21.5,4) -- (23.5,5);
            
            \pgfmathsetmacro{\x}{32}
            \pgfmathsetmacro{\y}{4}
            
            \fill[fill=gray!70!white] (\x,\y) -- (\x,\y+5.4) -- (\x+1,\y+5.4) -- (\x+1,\y+1) -- (\x+3,\y+1) -- (\x+3,\y+3) -- (\x+5,\y+3) -- (\x+5,\y+1) -- (\x+5.4,\y+1) -- (\x+5.4,\y) -- (\x,\y);
            \draw (\x+1,\y+5.4) -- (\x+1,\y+1) -- (\x+3,\y+1) -- (\x+3,\y+3) -- (\x+5,\y+3) -- (\x+5,\y+1) -- (\x+5.4,\y+1);
            
            \filldraw[fill=red!50!white] (\x+1,\y+3) -- (\x+1,\y+5) -- (\x+3,\y+5) -- (\x+3,\y+3) -- (\x+1,\y+3);
            \filldraw[fill=red!50!white] (\x+3,\y+3) -- (\x+3,\y+5) -- (\x+5,\y+5) -- (\x+5,\y+3) -- (\x+3,\y+3);
            \draw[draw=red, line width=0.5mm, -latex] (29.5,7) -- (31.5,7);
            
            \draw (15,-4) -- (15,10);
            
        \end{tikzpicture}
    \end{center}
    \caption{(left) example of a situation in which the growth of a duple tile can interfere with the growth of other tiles even at a distance from where the tile grew. (right) This same behaviour implemented using geometric blocking. Notice that even though only half of the duple can grow at any time, the geometry presented by one half of the duple is designed to be incompatible with the normal blue tile geometry so only the remaining half can grow in that location. Also notice that the geometry presented when the blue tile grows is incompatible with the first half of the duple.}
    \label{fig:duple-interact}
\end{figure}
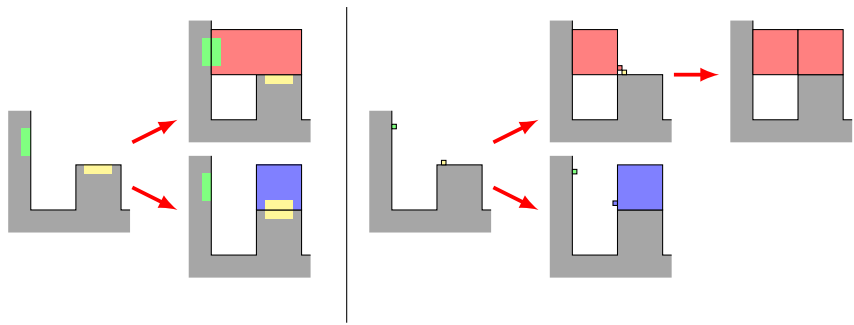

The tiles in $\calS$ are constructed exactly like in the aTAM simulation construction; there will be $16$ tiles corresponding to each non-duple tile $t$ in $\calD$ with each glue of $t$ being replaced by on of the two corresponding geometries and glues. The duple tiles are constructed similarly. For each duple tile $d$ in $\calD$ there will be $16$ corresponding tiles in $\calS$, with $8$ tiles in $\calS$ representing each half of the duple. The reason there are $8$ corresponding tiles for each half is because one of the glues, namely the one that represents the center of the duple, only has an $\alpha$ version. The non-center glues are replaced by one of the two corresponding geometries just like the normal tiles. Furthermore, we construct the seed configuration identically to the aTAM simulation construction, noting that the glue between two duple tile halves will never mismatch.

We've already demonstrated that the dynamics of the normal tiles will be simulated by $\calS$ since they are identical to the previous construction. It's also straightforward to see that if either half of a duple tile attaches, no normal tile will ever be able to grow into the location corresponding to the other half since the normal flag and duple flag are incompatible bumps. Furthermore, only the correct second half of a duple will grow after the first one since each of the duple geometries are only compatible with themselves and each half of the duple has the duple glue on an opposite face. Lastly, if any tile where to grow into a location that would belong to the second half of a duple tile before the first half grew, it would be impossible for the first half of the duple to grow since the geometries would mismatch. This behaviour emulates that of the duple system and therefore $\calS$ simulates $\calD$. \qed
\end{proof}

\subsection{Proof of Theorem~\ref{thm:imposs}}

\begin{proof}To prove Theorem~\ref{thm:imposs}, we first present the details of the aTAM SASS $\calS$, and then explain why no GTAM system at temperature-1 can simulate it.
The temperature-2 aTAM SASS system $\calS$ grows as follows. (Reference Figure~\ref{fig:imposs-overview} for a graphical depiction.) From the seed tile (shown as a black square on the far left), a module called the $\planter$ (similar to the $\planter$ module introduced in \cite{jCCSA}) grows horizontally to the right (white portion). At a high level, the $\planter$ is a binary counter which counts from $4$ to infinity, with spacing between each increment that allows the bits of the current binary number to be rotated to the north, as well as $6$ additional spaces. At each location where the bits of a binary number $b$ are exposed to the north, a binary $\decrementer$ grows upward (light grey rectangles in Figure~\ref{fig:imposs-overview}) to a height of $2b$ by having each successive pair of rows decrement the values from $b-1$ to $0$, then check to see if $0$ has been reached. Once the $\decrementer$ reaches $0$ and terminates upward growth, a row of tiles grows east from its top row, a distance of $4$ tiles. Then, from the easternmost tile of that row, a column of tiles (green in Figure~\ref{fig:imposs-overview}) grows downward until being blocked by the $\planter$ (shown in more detail in Figure~\ref{fig:imposs-mid}). At that point, cooperation between the final green tile and a tile in the $\planter$ allow for placement of the yellow tile, which then exposes a strength-2 glue allowing the red tile to attach. We call the growth of a $\decrementer$ through the placement of its associated red tile an $\iteration$, and note that the infinite terminal assembly of $\calS$ grows exactly one $\iteration$ for every number $\ge 4$. (Note that the counter within the $\planter$, the $\decrementer$ module, etc. all use standard, basic aTAM modules.)



Note that, as depicted in Figure~\ref{fig:imposs-mid}, the growth of the modules of an $\iteration$ in general follows a zig-zag pattern which allows the growth to be completely sequential, only ever having a single frontier location, and thus maintaining the property that $\calS$ is a SASS. Additionally, the growth of the $\planter$ extends $5$ columns beyond the east side of its associated $\decrementer$ before backtracking and growing the top row, which ensures that the tiles that block and cooperate to place the green and yellow tiles, respectively, are in place before the green column grows downward. Thus, clearly, $\calS$ is directed, and since at every point during its growth there is exactly one single frontier location, $\calS$ is a SASS.



We prove Theorem~\ref{thm:imposs} by contradiction, and therefore assume that some temperature-1 GTAM system $\calG$ correctly simulates $\calS$, and let $m$ be the scale factor of that simulation. Without loss of generality, we will assume that the tiles of $\calG$ contain no glues to which no other tile can bind, as any such glues could be removed without changing the behavior of $\calG$. Now, we note that the height of each successive $\decrementer$, and therefore the heights of the associated green columns, increases infinitely. However, $\calG$ must have a tile set with a finite number of glue types. Let $g$ be the number of glue types in $\calG$ and for our proof we'll let $n = ((g+1)^{6m}\cdot(6m)!+1)\cdot3+2$ (we'll discuss the selection of this value later) 
and we'll focus on the growth of an assembly produced by $\calG$ up to and including the first tile placement of $\iteration$ $n$ in which the red tile of $\iteration$ n of $\calS$ is represented (depicted in Figure~\ref{fig:imposs-zoom} as $b$). Let $\alpha$ be the assembly created at this point, and note that by the assumption that $\calG$ simulates $\calS$, such an $\alpha$ must be producible in $\calG$. Furthermore, we will define $\vec{\alpha} = (\alpha_i$ $|$ $0 \le i < k)$ to be the assembly sequence which produces $\alpha$ and give it the following restriction: whenever there are multiple frontier locations in any $\alpha_i$ (since $\calG$ itself need not be a SASS even though $\calS$ is), the next location to receive a tile will be selected from those with the lowest $y$-coordinate. This would clearly be a valid assembly sequence, and since $\calG$ is simulating $\calS$, which is a SASS and grows each portion of each $\iteration$ in the previously specified order, it must be the case that no locations further than the single supertile boundary of fuzz outside of the first $n$ $\iteration$s can receive a tile before the $n$th $\iteration$ completes. Additionally, selection of the assembly sequence ensures that all tiles which share a path in the binding graph (i.e. a path through a series of bound glues) that includes a tile in the $\planter$ without going through the full $\decrementer$ and green column (i.e. those which may have grown upward as fuzz from the $\planter$), will be placed before the second row of the $\decrementer$ completes growth.

The tile placement which causes the red tile of $\iteration$ $n$ to be represented (shown as $b$ in Figure~\ref{fig:imposs-zoom}) must be placed only after representations of the entire $\decrementer$, the green column, and the yellow tiles have grown. This means that it must have a path in the binding graph which connects upward through the supertiles representing the green column, since the assembly sequence $\vec{\alpha}$ was chosen to place any tiles which didn't have such a path before the $\decrementer$ completes, and if a tile placement which caused the red tile to be represented were placed before growth of the green column, then $\calG$ would not correctly simulate $\calS$ (which can only place the red tile after the $\decrementer$ and green column complete). We now know that there must be a path of bound glues connecting the tile at location $b$ to some tile within the green column (which we label as $a$ in Figure~\ref{fig:imposs-zoom}).

We will now consider cuts which separate the green column in $\alpha$ into two pieces and cross between macrotiles (of width $m$) representing a green tile, as well as possibly one macrotile of allowable fuzz on each side, for a maximum width of $3m$. The count of the maximum possible number of glue positions along that cut from tiles both above and below it is then $6m$. With $g$ glue types and the empty glue, there are $(g+1)^{6m}$ different ways to fill those positions. The total possible number of orderings of placements at each of those $6m$ positions is $(6m)!$, and therefore $(g+1)^{6m}\cdot(6m)! + 1$ is greater than the total possible variety of glues and orderings of their placements along the cut. The quantity $n$ represents the number of the final iteration of $\calS$ which is being represented in $\calG$, and thus the height of the column of green tiles in $\calS$, and macrotiles representing those tiles in $\calG$, will be $n+4$. By setting $n = ((g+1)^{6m}\cdot(6m)! + 1)\cdot 3 + 2$, we can ignore the top and bottom three green macrotiles and still guarantee that even if we only inspect the cuts between every third pair of macrotiles, that at least two of those cuts have the same window movie (i.e. they had the exact same glues placed in the exact same order). Now, using those two cuts we define the windows $w$ and $w'$ (exemplified in Figure~\ref{fig:imposs-windows}) such that they each have one of those identical cuts as their northern horizontal cut, their western edge travels down between any allowable fuzz macrotiles (and there must be at least one completely empty macrotile space between the $\decrementer$ and the green column or the fuzz boundary would be broken), and then back to the east in a perhaps jagged path which avoids crossing any bound glues but separates the portion of $\alpha$ attached to the green column from the $\planter$. This type of cut must be possible because, as previously mentioned, the assembly sequence $\vec{\alpha}$ was selected to first place all tiles which could be bound by glues through a path directly to the $\planter$ (avoiding the green column) before the green column forms, and since $\calG$ is temperature-1, if tiles could bind after growing down from the green column, they could also bind before it forms. Thus, the boundary of the assembly on the north side of the $\planter$ which forms before the green column can be followed for this cut, to finish the windows.

Note that $w'$ is not a full translation of the window $w$ by some vector $\vec{c}$, but instead only the upper horizontal cut which contains all locations which could have a glue bond (all other locations along the window cut were chosen so that no glue bonds were formed across them). By Lemma 3.3 (the Window Movie Lemma) of \cite{IUNeedsCoop}\footnote{We'll refer directly to the arXiv version here \url{https://arxiv.org/pdf/1304.1679.pdf} for convenience}, and more specifically Corollary 3.4 which refers only to the bond-forming submovies, it must be the case that the assembly depicted on the right of Figure~\ref{fig:imposs-windows} must be able to form. Essentially, the existence of two cuts across the green column where identical series of glues are placed along those cuts allows the segment in between those cuts to be ``pumped'', potentially either upward (increasing the number of occurrences of the subassembly in between) or downward (decreasing them). Here we choose to pump down and remove the intermediate subassembly, showing that the path of tiles which grows to the right from the green column must be able to also grow at a higher location.  Since, as we have shown, this path must extend greater than the width of a single supertile beyond the green column (i.e. through the width of the yellow macrotile and into the red), it is beyond the allowable fuzz region when translated upward, and therefore $\calG$ does not correctly simulate $\calS$. This is a contradiction that $\calG$ simulates $\calS$, and therefore Theorem~\ref{thm:imposs} is proven. $\qed$

\end{proof}

\subsection{Proof of Theorem~\ref{thm:imposs-general}}

\begin{proof}
The proof of Theorem~\ref{thm:imposs-general} follows almost immediately by using an identical proof structure to the proof of Theorem~\ref{thm:imposs}. Let $\calS$ be the same as $\calS$ defined in that proof. Again, our proof will be by contradiction, and we will therefore assume that $\calW$ is a weakly-cooperative tile assembly system which utilizes geometric hindrance and simulates $\calS$. By the definition of simulation, all of the same fundamental restrictions that held for $\calG$ of the former proof also hold for $\calW$, regardless of the model in which it is contained. There must be a bound on the number of glues $g$ (and tile types), and a regular grid of macrotiles bound to some scale factor $m$, and we can similarly compute a bound on the number of window movies which are possible across the boundaries of the macrotiles representing the green tiles of $\calS$. In some other models the tiles may be larger than single unit squares (e.g. polyominoes\cite{Polyominoes}), or they may consist of multiple shapes (e.g. squares and duples\cite{jDuples}), or they may be able to meet at differing angles (e.g. polygons\cite{Polygons}),  and in such systems the tiles binding across macrotile boundaries may not form straight lines along the macrotile boundaries, but instead may create jagged boundaries. However, even though this may increase the number of window movies possible, they are still finitely bounded. Therefore, we can once again find a pair of windows with identical bond-forming submovies, as well as a path in the binding graph from the green macrotiles to the red. Once again, we can use the Window Movie Lemma to pump down between the top cuts of the windows, and $\calW$ must grow that path at a higher location, therefore breaking the boundary allowed for fuzz. This contradicts the assumption that $\calW$ simulates $\calS$, and therefore Theorem~\ref{thm:imposs-general} is proven. $\qed$
\end{proof}

\end{document}